\documentclass[11pt]{article}
\usepackage{authblk}
\usepackage[a4paper,margin=3cm]{geometry}
\usepackage[utf8]{inputenc}
\usepackage{amsmath}
\usepackage{amssymb}
\usepackage{amsthm}
\usepackage{dsfont}
\usepackage{color}
\usepackage{amsmath}
\usepackage{amsfonts}
\usepackage{amssymb}
\usepackage{bbm}
\usepackage{lineno}
\usepackage{relsize}
\usepackage{enumitem}
\usepackage{pst-node}
\usepackage[procnumbered, linesnumbered,algoruled]{algorithm2e}
\usepackage[compact]{titlesec}
\usepackage{tikz}
\usetikzlibrary{topaths,calc}
\usetikzlibrary{positioning}
\usetikzlibrary {shapes,matrix}
\usetikzlibrary{arrows}
\tikzset{
	semithick,
	node distance = 2cm,
	dot/.style={circle,fill,inner sep=2pt}
}
\tikzset{
	side by side/.style 2 args={
		line width=2pt,
		#1,
		postaction={
			clip,postaction={draw,#2}
		}
	}
}
\tikzstyle{every state}=[draw = black,thick,fill = white,minimum size = 4mm]
\tikzstyle{selected edge} = [draw,line width=2pt,-,red!50]
\tikzset{
	edge/.style={->,> = latex'}
}
\usepackage[colorlinks=true,linkcolor=red,citecolor=blue]{hyperref}
\usepackage[nameinlink]{cleveref}

\newcommand{\comment}[1]{}
\usepackage[colorlinks=true,linkcolor=red,citecolor=blue]{hyperref}

\usepackage[nameinlink]{cleveref}
\usepackage[procnumbered, linesnumbered,algoruled]{algorithm2e}

\newcommand{\op}{{\mathbb{I}_{p,t}}}
\newcommand{\cA}{{\mathcal{A}}}

\newcommand{\cG}{{\mathcal{G}}}

\newcommand{\cF}{{\mathcal{F}}}
\newcommand{\cC}{{\mathcal{C}}}
\newcommand{\cP}{{\mathcal{P}}}
\newcommand{\cS}{{\mathcal{S}}}

\newcommand{\OPT}{\textnormal{OPT}}

\newcommand{\E}{{\mathbb{E}}}

\newcommand{\om}{{\Omega}}
\newcommand{\cct}{{\cC(I_t)}}
	\Crefname{obs}{Observation}{Observations}
		\Crefname{cor}{Corollary}{Corollaries}
\newcommand{\bbc}{{\bar{\cC}(I_t)}}
\begin{document}
\newtheorem{thm}{Theorem}[section]
\newtheorem{prop}[thm]{Proposition}
\newtheorem{assm}[thm]{Assumption}
\newtheorem{lem}[thm]{Lemma}
\newtheorem{obs}[thm]{Observation}
\newtheorem{cor}[thm]{Corollary}
 \newtheorem{lemma}[thm]{Lemma}
  \newtheorem{definition}[thm]{Definition}
 \newtheorem{theorem}[thm]{Theorem}
 \newtheorem{proposition}[thm]{Proposition}
 \newtheorem{claim}[thm]{Claim}
\newtheorem{defn}[thm]{Definition}
\newcommand{\ariel}[1]{{\color{red} (Ariel :#1)}}
\def \II   {{\mathcal I}}
\newcommand{\one}{\mathbbm{1}}
	\def\claimproof{\proof}
\def\endclaimproof{\hfill$\square$\\}
\renewcommand\qedsymbol{$\blacksquare$}

\title{ {\bf Non-Linear Paging}}

%\author{Anonymous Submission ICALP 24}

\author{Ilan Doron-Arad \and Joseph (Seffi) Naor} %\thanks{Computer Science Department, 
%	Technion, Haifa, Israel, 209007939.\texttt{idoron-arad@cs.technion.ac.il}}}

%\and	
%	Ariel Kulik\thanks{CISPA Helmholtz Center for Information Security, Germany. \texttt{ariel.kulik@cispa.de}} 
%\and 
%Hadas Shachnai\thanks{Computer Science Department, 
%Technion, Haifa, Israel. \texttt{hadas@cs.technion.ac.il}}

\maketitle

\begin{abstract}
We formulate and study {\em non-linear paging} - a broad model of online paging where the size of subsets of pages is determined by a monotone non-linear set function of the pages. This model captures the well-studied classic weighted paging and generalized paging problems, and also {\em submodular} and {\em supermodular} paging, studied here for the first time, that have a range of applications from virtual memory to machine learning. 

Unlike classic paging, the cache threshold parameter $k$ does not yield good competitive ratios for non-linear paging. Instead, we introduce a novel parameter $\ell$ that generalizes the notion of cache size to the non-linear setting. We obtain a tight deterministic $\ell$-competitive algorithm for general non-linear paging and a $o\left(\log^2 (\ell)\right)$-competitive lower bound for randomized algorithms. Our algorithm is based on a new generic LP for the problem that captures both submodular and supermodular paging, in contrast to LPs used for submodular cover settings. 
We finally focus on the supermodular paging problem, which is a variant of online set cover and online submodular cover, where sets are repeatedly requested to be removed from the cover. We obtain  polylogarithmic lower and upper bounds and an offline approximation algorithm.  

% with a {\em non-linear}  with an arbitrary {\em non-linear} function
\end{abstract}

\section{Introduction}
In the well studied {\em paging} problem, we are given a collection of $n$ pages and a cache that can contain up to $k$ pages simultaneously, where $k<n$. At each time step, one of the pages is requested. If the requested page is already in the cache, the request is immediately served. Otherwise, there is a {\em cache miss} and the requested page is fetched to the cache; to ensure that the cache contains at most $k$ pages, some other page is potentially evicted. In the most fundamental model, the goal is to minimize the number of cache misses (or equivalently, number of evictions).

\comment{
In more general models, pages may have different sizes and (eviction) costs (see, e.g., \cite{adamaszek2018log,bansal2012randomized}) and then the sum of the sizes of the pages in the cache cannot exceed the cache size. However, a linear function over the page sizes that defines cache feasibility fails to capture scenarios with more involved relations between subsets of pages that can reside together in the cache. For example, consider a system in which pages share parts of their memory; here, only the missing memory parts of the requested page contribute to the increment of the cache size. This setting can be modeled using a {\em submodular} function to define cache feasibility. 
In a second example, items stored in cache may have running-time dependencies inflicting additional overhead in their mutual storage demand. Thus, the size of a page may increase when certain other pages are in the cache. {\bf Seffi: do you have a reference?.} 
A {\em supermodular} size function can describe cache feasibility in this case. These are two out of many examples for cache systems that exhibit a {\em non-linear} behavior. We motivate our model in Section in detail \ref{sec:motivation}.  
%{\bf Seffi: Maybe we should add another example? Rule caching?.} (e.g., \cite{lilja1993cache, bennett1990adaptive, wilson1987hierarchical, suh2004dynamic})
}
In more general models, pages may have different sizes and (eviction) costs (see, e.g., \cite{adamaszek2018log,bansal2012randomized}) and then the sum of the sizes of the pages in cache cannot exceed its capacity. However, a linear function over page sizes that defines cache feasibility fails to capture scenarios with more involved relations between subsets of pages that can reside together in cache. %Two examples of cache systems that exhibit {\em non-linear} behavior are the following.
Consider a system in which pages share parts of their memory and then only the missing memory parts of a requested page can contribute to the increase in cache size. This setting can be modeled using a {\em submodular} function to define cache feasibility.
Another example is a setting in which items stored in cache have dependencies yielding additional overhead in their mutual storage demand. The {\em rule caching} problem is one such setting and it has been well studied in networking \cite{dong2015rule,yan2014cab,sheu2016wildcard,huang2015cost, li2019tale,stonebraker1990rules,li2015fdrc,gao2021ovs,cheng2018switch,rottenstreich2016optimal,li2020taming,gamage2012high,yan2018adaptive,rottenstreich2020cooperative,rastegar2020rule,yang2020pipecache,doron2024approximations}. %Thus, the size of a page may increase when certain other pages are in the cache. {\bf Seffi: do you have a reference?.} 
%A {\em supermodular} size function can describe cache feasibility in this case. 
 
 Our focus in this paper will be on settings where there is such non-linear behavior.
 We will further motivate our model in detail in \Cref{sec:motivation}.  

\subsection{Our Model}

%\subsubsection*{Preliminaries}
Before presenting our model, we give some required preliminary definitions. % In addition, for any $w \in \mathbb{N}$, let $[w] = \{1,2,\ldots,w\}$. 
Let $\cP$ be a set and let $f:2^{\cP} \rightarrow \mathbb{N}$ be a set function of $\cP$. The function $f$ is called {\em monotone} if for every $S \subseteq \cP$ and $S'\subseteq S$ it holds that $f(S') \leq S$. In addition, $f$ is called {\em submodular} if for every $S,S' \subseteq \cP$ it holds that $f(S)+f(S') \geq f\left(S \cup S'\right)+f\left(S \cap S'\right)$. Conversely, $f$ is called {\em supermodular} if for every $S,S' \subseteq \cP$ it holds that $f(S)+f(S') \leq f\left(S \cup S'\right)+f\left(S \cap S'\right)$.

In this work, we introduce a very general model of paging with an arbitrary function defining cache feasibility. In the {\em non-linear paging} problem, we are given a collection $\cP$ of $n$ pages where each page $p \in \cP$ has a fixed eviction cost $c(p)$. We are also given a monotone {\em feasibility} function $f:2^{\cP} \rightarrow \mathbb{N}$ that assigns a value to every subset of pages, indicating their size. Finally, we are given a cache threshold $k$. As in standard paging, in each time step $t$ there is a request $p_t$ for one of the pages. If $p_t$ is already in the cache, the request is immediately served. Otherwise, $p_t$ is fetched to the cache and possibly some subset of pages is evicted to ensure that the set of pages in the cache, denoted by $S_{t}$, is feasible, i.e., $f\left( S_{t}  \right) \leq k$. The goal is to minimize the total cost incurred from page evictions. 
%{\bf Seffi: We may have discussed this in the past. Doesn't it make more sense to denote by $S_t$ the cache contents after serving request $p_t$? Then $S_0$ is the initial contents, $S_1$ after serving $p_1$, etc.}%cache misses. 
%
%We refer to this model as {\em non-linear paging}, even though it encompasses linear feasibility functions as a special case. Namely, 
The classic paging problem is obtained by setting $f(S) = |S|$ for all $S \subseteq \cP$. Other interesting applications of our model are described below. 

\begin{itemize}
\item  {\em Generalized Paging} \cite{bansal2012randomized,adamaszek2018log}. Here the feasibility function is linear; that is, for every $S \subseteq \cP$ it holds that $f(S) = \sum_{p \in S} f(p)$, where $f(p)$ is the size of page $p \in \cP$.
 %$f(p) \in \{1,2,\ldots,k\}$ for all $p \in \cP$.
\item 	{\em Submodular Paging}. The feasibility function is {\em submodular}. A natural application of this variant is to settings where pages share memory items (see \Cref{sec:motivation}).

	%  Consider a cache that can store up to $k$ {\em atoms} from a larger set $A = \{a_1,\ldots, a_n\}$. Each page $p$ in the instance has a corresponding subset of atoms $a(p) \subseteq A$. The feasibility function ensures that a collection $S$ of pages contains jointly at most $k$ atoms: $f(S) = \left| \bigcup_{p \in S} a(p) \right|$, which is a submodular function.   

\item 	{\em Supermodular} Paging. The feasibility function is {\em supermodular}, implying a submodular cover function for pages remaining out of the cache. This setting effectively captures online submodular covering problems \cite{alon2003online,korman2004use,gupta2020online}. Supermodular paging will be a main focus of our paper.  
\end{itemize}

%   , making the supermodular paging problem a generalization of the well-known {\em set cover} and {\em submodular cover} \cite{wolsey1982analysis} (we give more details in \Cref{sec:hardness}). %\cite{alon2003online}. 
   
Supermodular paging is a variant of {\em online set cover} \cite{alon2003online} and {\em online submodular cover}~\cite{gupta2020online}. In online set cover, we are given  a ground set $X$ and a family $\cS$ of subsets of $X$. Requests for  elements of $X$ arrive online; if a requested element is not already covered by a previously chosen set,  a set $S \in \cS$ containing it is chosen, paying a cost $c(S)$. The goal is to minimize the cost of selected sets. Online submodular cover generalizes online set cover - the goal is to cover a general (monotone) submodular function with an increasing cover demand over time. %{\bf Seffi: the next sentences on supermodular paging are not clear. What are the "above problems"? (even more challenging than the above problems).} 
In supermodular paging (submodular cover), the cover demand does not change over time, as the same page (a set $S \in \cS$ in online set cover) may be requested (removed from the cover) multiple times. We later show (see \Cref{sec:hardness}) that supermodular paging is even more challenging than (online) submodular cover. %the above problems. 

   \subsection{Motivation}
   \label{sec:motivation}
   
   Non-linear paging generalizes several fundamental caching problems, capturing many real world applications. Besides known applications of the classic paging models with linear feasibility functions \cite{sleator1985amortized,fiat1991competitive,bansal2012primal,bansal2012randomized,adamaszek2018log}, there are many scenarios in which the interaction between pages is non-linear, requiring the more general non-linear paging model. 
   %As we show in \Cref{sec:hardness}, 
   As already indicated, supermodular paging roughly generalizes online set cover and therefore has both theoretical and practical importance \cite{gupta2020online,alon2003online}. We describe below several interesting applications of non-linear paging.

A major motivation for the study of non-linear paging %Another setting in which non-linear paging is valuable, in particular submodular paging, 
is caching in shared memory systems (e.g., \cite{lilja1993cache, bennett1990adaptive, wilson1987hierarchical, suh2004dynamic}).  Each process in a multi-process memory system is associated with a {\em virtual memory} \cite{denning1970virtual}, providing the illusion that it has a much larger memory.  In some shared memory systems with multiple processes, caching policies are defined over entire processes, that is, the entire virtual address space of a process can be taken to the cache. Since the virtual memory of two processes typically overlaps in physical memory, the increase in (physical) memory in cache is larger if the cache is empty, as the entire memory of a process is loaded to the cache. In contrast, if the cache is nearly full, and a process is loaded to the cache, the increase in total size is smaller, as most of the virtual memory addresses are already in cache. In a similar vein, when caching hot data at the network edge, to avoid serving requests from a remote cloud, space efficiency of similar files is achieved through {\em deduplication} (e.g. \cite{LiL20}), which is very similar to the way overlaps are handled in virtual memory.
%In this setting, space efficiency together with data popularity and future request rates need to be taken into account, see e.g. \cite{LiL20}.
   
Thus, caching in shared memory systems is a special case of submodular paging. More formally, consider a cache that can store up to $k$ {\em atoms} from a larger set $A = \{a_1,\ldots, a_n\}$ spanning the physical memory. Each page (process) $p$ contains a subset of atoms $a(p) \subseteq A$ corresponding to the physical memory to which process $p$ is mapped. The feasibility function $f$ ensures that a collection $S$ of pages contains jointly at most $k$ atoms. Thus, $f(S) = \left| \bigcup_{p \in S} a(p) \right|$, and it is a submodular function. See \Cref{fig:X} for an illustration. %of the above. 

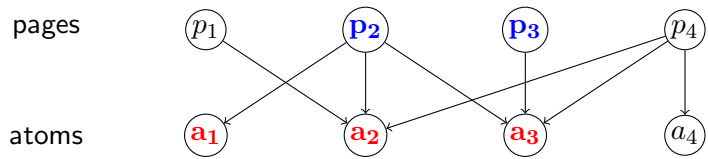
\begin{figure}
  	%	\hspace{4cm}{
  		\centering
  		\begin{tikzpicture}[scale=1.4, every node/.style={draw, circle, inner sep=1pt}]
  			% first bipartite graph
  			\node (p2) at (5.5,-0.5) {$\bf \textcolor{blue}{p_2}$};
  			\node (p1) at (4,-0.5) {$\textcolor{black}{p_1}$};
  			\node (p3) at (7,-0.5) {$\bf \textcolor{blue}{p_3}$};
  			\node (a1) at (4,-1.5) {$\bf \textcolor{red}{a_1}$};
  			
  			\node (a2) at (5.5,-1.5) {$\bf \textcolor{red}{a_2}$};
  			\node (a3) at (7,-1.5) {$\bf \textcolor{red}{a_3}$};
  				\node (a4) at (8.5,-1.5) {$a_4$};
  				\node (p4) at (8.5,-0.5) {$p_4$};
  			\draw[->] (p2) -- (a2);
  			%		\draw (p3) -- (a2);
  		%	\draw[line width=2pt, color=red] (p3) -- (a3);
  			
  			\draw[->] (p3) -- (a3);
  			\draw[->] (p2) -- (a1);
  			\draw[->] (p1) -- (a2);
  			\draw[->] (p2) -- (a3);
  			
  			\draw[->] (p4) -- (a2);
  			\draw[->] (p4) -- (a3);
  			\draw[->] (p4) -- (a4);

  			\node[draw=none] at (2.5, -0.5) {$\textsf{pages}$};
  			
  			\node[draw=none] at (2.5, -1.5) {$\textsf{atoms}$};
  			
  			%			\node[draw=none] at (10.5, -2) {$H = (L,R,\bar{E})$};
  			%			
  			%			\node[draw=none] at (10, -1.5) {$d(i_1) \leq d(i_{B_2})$};
  			%%			
  			%			\node[draw=none] at (10, -0.5) {$w(\ell_i) = 20$};
  			%%			
  			%			\node[draw=none] at (10, -1) {$B = 74$};
  		\end{tikzpicture}
  		%\vspace{-1.5cm} 
  		\caption{\label{fig:X} An illustration of a shared memory systems with four pages and four shared memory units (``atoms"). The cache is of size $k = 3$ and contains pages $p_2,p_3$ (in blue) whose shared memory is of size $3$ - the three red atoms. Observe that either fetching $p_1$ or evicting it will not change the number of atoms stored in cache, however, fetching $p_4$ requires evicting $p_2$.}
  \end{figure}

Paging with a supermodular feasibility function arises in common settings where the storage demand grows rapidly as a function of the number of ``pages" stored in cache. In these scenarios, there is a large set of $n$ {\em entities} (corresponding to vertices), and among subsets of entities certain interactions exist (represented by hyperedges). This data structure, known as a hypergraph, is ubiquitous in applications such as recommendation systems \cite{ghoshal2009random,tan2011using} (where vertices represent individuals and hyperedges represent communities), image retrieval \cite{liu2011hypergraph}(vertices represent images and hyperedges represent correlations), and bioinformatics \cite{patro2013predicting} (vertices represent substances and hyperedges stand for biochemical interactions). Other applications arise in machine learning \cite{tian2009hypergraph,gao2012visual,zhou2006learning} and databases \cite{beeri1983desirability}.

In various practical settings, the hypergraph is very large (e.g., \cite{ruggeri2023community,murakami2013efficient,huang2015scalable}). Therefore, a natural approach is to store frequently accessed vertices in a cache. However, caching is effective only if all interactions (i.e., hyperedges) among subsets of vertices in cache are also stored therein. As a set of $x$ vertices may have up to $2^x$ induced hyperedges, the storage demand for hyperedges tends to be significantly larger compared to the number of vertices. Caching hypergraphs in the non-linear paging framework can be formally defined as follows: consider a set $\mathcal{P}$ of vertices and define a feasibility function $f$ over $\mathcal{P}$, such that for a subset of vertices $S \subseteq \mathcal{P}$, the storage demand $f(S)$ is the number of induced hyperedges in $S$ plus the cardinality of $S$. Function $f$  is a supermodular function.

For example, a real-world problem related to supermodular paging is caching in device-to-
device (D2D) communication networks, with social ties among users and common interests that are used
as key factors in determining the caching policy and are modeled via a hypergraph \cite{bai2016caching}. Here, the number
of hyperedges (describing roughly interferences among users of the network, content, transmission rate,
etc.) grows in a supermodular manner w.r.t. the number of users placed in cache. In addition, our
reduction from online set cover (similarly, online submodular cover) to supermodular paging implies that
applications of online set cover are also applications of supermodular paging.

\subsection{Our Results and Techniques}

%A desirable property of online algorithm is that their competitive ratio is independent of the set of time steps, $T$. Our first motivation is to obtain such online algorithms for non-linear paging.
%strictly separates non-linear paging from the classic paging problem. 
%Crucially, 
We now present our results and elaborate on the techniques used. We start with general non-linear paging and then proceed to the special case of supermodular paging.

\subsubsection{General Non-Linear Paging}
\label{sec:ell}

In classic paging models, competitive ratios are typically given as a function of $k$, the cache size.  However, non-linear paging is more difficult. Even for non-linear paging instances with $k = 0$ the competitive ratio can be very high. For example, consider a (classic) paging instance $I'$ with cache threshold $k'$; define a non-linear paging instance $I$ with a (non-linear) feasibility function $f$ for which $f(S) = 0$ if $S$ is feasible for $I'$ and $f(S) = 1$ otherwise. In addition, we set $k = 0$ as the cache threshold of $I$. Clearly, a solution for $I$ implies a solution for $I'$. 
 Hence, by the hardness of paging \cite{sleator1985amortized,fiat1991competitive} the best competitive ratio of non-linear paging is arbitrarily large as a function of $k$ (we give the remaining details in \Cref{sec:hardRestricted}). %as a function of $n$. 
Thus, instead of $k$, we look for a parameter that better captures the competitiveness of general non-linear paging.
%but allows us to exhibit more understanding of the problem than the large factors of $\log (n) \cdot \log (f(\cP))$ as in \Cref{thm:randomized}.  
%We focus on a different parameter, which sheds more light on our problem. 
%The parameter that we use refers to 
This parameter turns out to be the maximum cardinality of a {\em minimally infeasible} set, i.e., an infeasible set where every proper subset of it is feasible\footnote{Technically, we subtract one so that the definition coincides with the parameter $k$ in paging.}. Formally, %set $S$ such that $f(S)>k$, 

\begin{definition}
	\label{def:Mininfeasible}
	A set $S \subseteq \cP$ is called {\em feasible} if $f(S) \leq k$ and is {\em infeasible} otherwise. Additionally, $S$ is {\em minimally  infeasible} if $S$ is infeasible and every $S' \subset S$ is feasible, and let $\mathcal{M} = \left\{S \subseteq \cP~|~f(S)>k \textnormal{ and } f(S') \leq k~\forall S' \subset S \right\}$ be all minimally infeasible sets. Finally, let the {\em width} of $f$ be $$\ell(f) = \max_{S \in \mathcal{M}} \left(|S|-1\right).$$ %as the maximum cardinality of a minimally infeasible set. %in the instance. 
	%	$$\ell = \max \left\{\right\}$$
\end{definition}
%We call $\ell$ the {\em width} of the instance. 
We simply let $\ell = \ell(f)$ when it is clear from context. 
Clearly, for paging (or weighted paging) the width $\ell$ equals $k$. Hence, the width accurately captures the optimal performance of paging algorithms: there is a tight $\ell$-competitive deterministic algorithm \cite{sleator1985amortized} and a tight $O(\log(\ell))$-competitive randomized algorithm \cite{fiat1991competitive}. However, the width behaves quite differently in other scenarios. For example, in the setting of submodular paging described in \Cref{sec:motivation}, the instance can have a fixed width $\ell = O(1)$ but the number of pages in cache can be unbounded (e.g., all pages use the same atom). Interestingly, we show that the width gives a tight competitive ratio also for general non-linear paging via a new LP for the problem (see \Cref{sec:LP,sec:algUneweighted}).   %$\ell$-competitive deterministic algorithm  

\begin{theorem}
	\label{thm:deterministic}
	There is a deterministic $\ell$-competitive algorithm for \textnormal{non-linear paging}. Moreover, every deterministic algorithm for  \textnormal{non-linear paging} is at least $\ell$-competitive. 
\end{theorem}
%\subsection{}

The algorithm achieving \Cref{thm:deterministic} is based on a new LP relaxation, designed for the general setting of non-linear paging. We now explain why previous techniques for classic paging are not sufficient for obtaining a competitive online algorithm for this setting.

Previous work on generalized paging \cite{bansal2012randomized,adamaszek2018log} and on other online covering problems \cite{gupta2020online,coester2022competitive} use  {\em knapsack cover} constraints. 
This powerful technique, originated by Wolsey \cite{wolsey1982analysis}, is very useful for relaxing submodular cover constraints
using linear inequalities. 
%for offline submodular cover, reduces the integrality gap of the relaxation and enables the design of online algorithms with small competitive ratios. We provide below the necessary definitions. %necessary definition of {\em covering} function in this context. 
Indeed, paging problems are often studied from the viewpoint of a {\em covering} problem (e.g., \cite{bansal2012primal}), where the complement of the cache (i.e., pages outside the cache) needs to be covered. In classic paging, at any point of time at least $n-k$ pages are not in the cache. %To generalize this concept for non-linear paging, 

Consider the {\em covering} function $g:2^{\cP} \rightarrow \mathbb{N}$ of a feasibility function $f$ defined to be the (non-linear) size requirement {\em outside} of the cache: $g(S) = f(\cP)-f(\cP \setminus S)$ for every $S \subseteq \cP$. Observe that for classic paging it holds that $g(S) = n-(n-|S|) = |S|$ and the feasibility constraint translates to $g(S) \geq n-k$ at all times. If $f$ is submodular (i.e., submodular paging) then $g$ is supermodular, and vice versa. %if $f$ is supermodular, then $g$ is submodular. 
%In our notation, given a variable $x_p$ indicating the contribution of $p$ to the cover, knapsack constraints require that every subset $S \subseteq \cP$ satisfy that $\sum_{p \in S} x_p \cdot g_S(\{p\}) \geq g_S(\cP)$.\footnote{For $S,S'\subseteq \cP$ %and $p \in \cP$ we use the notation $g_S(S') = g(S' \cup S)-g(S)$.}  
Knapsack cover constraints yield a relaxation of the covering problem when the cover function $g$ is submodular \cite{gupta2020online}, as is the case in classic paging problems.
%. Hence, for supermodular paging, in which the induced covering function $g$ is submodular, we may use the classic knapsack cover constraint. 

However, for submodular paging, the covering function $g$ is supermodular and knapsack cover constraints do not even provide a relaxation of the problem.  For example, let $g(S) = 1$ for $S = \cP$ and $g(S) = 0$ otherwise. Then, the knapsack constraints are not satisfied by the unique solution that covers a demand of $k = 1$ (the entire set $\cP$). Specifically, $\cP$ does not satisfy the knapsack constraint for $S = \emptyset$, i.e.,  $\sum_{p \in \cP} x_p \cdot g_{\emptyset}(\{p\}) = 0$, but $g_{\emptyset}(\cP) = 1$. 

To circumvent the limits of knapsack cover constraints for submodular paging, we formulate a new set of covering constraints that are valid for any feasibility functions $f$ and $g$. %giving a tight competitive ratio for deterministic algorithms for non-linear paging.
Specifically, the constraints require removing at least one page from every infeasible set. Then, using the online primal-dual approach applied to this set of constraints, we obtain a tight $\ell$-competitive deterministic algorithm for non-linear paging. Specifically, upon arrival of a page that induces an infeasible set of pages in cache, our algorithm identifies a minimally infeasible set of pages and continuously increases their corresponding dual variable in the LP, evicting {\em tight} pages.  %This turns out to be a relaxation  

Interestingly, as a special case, \Cref{thm:deterministic} gives a simple $k$-competitive deterministic algorithm for generalized paging; to the best of our knowledge, there are only $(k+1)$-competitive deterministic algorithms \cite{cao1997cost, young2002line} for generalized paging. Thus, our bound is tight for this problem.
%Since paging is a special case, and cannot have a better than a $k$-competitive deterministic algorithm \cite{sleator1985amortized}, our algorithm gives a tight competitive ratio for this problem. 

\begin{cor}
	\label{thm:generalized}
	There is a deterministic $k$-competitive algorithm for \textnormal{generalized paging}. %Moreover, every deterministic algorithm for  \textnormal{non-linear paging} is at least $\ell$-competitive. 
\end{cor}

We emphasize that the lower bound of \Cref{thm:deterministic} can be obtained for {\em any} function $f$ with a minimally infeasible set of cardinality $\ell$ (regardless of whether $f$ is linear, submodular, supermodular, or any other function). This shows the robustness of the parameter $\ell$ as an indicator for the competitiveness of non-linear paging. Thus, it is natural to ask whether the parameter $\ell$ for non-linear paging is analogous to the parameter $k$ for classic paging when allowing randomization. We answer this question in the negative by showing that in contrast to paging, that admits an $O(\log(\ell))$-competitive randomized algorithm \cite{fiat1991competitive}, non-linear paging is substantially harder w.r.t. the parameter $\ell$ (see \Cref{sec:hardness}).

\begin{theorem}
	\label{thm:LB}
Unless \textnormal{NP $\subseteq $ BPP}, there is no polynomial-time randomized $o(\log^2 (\ell))$-competitive algorithm for \textnormal{non-linear paging}.
\end{theorem}

A very intriguing open question is whether there exists a randomized $\textnormal{polylog} (\ell)$-competitive algorithm for general non-linear paging. Unfortunately, we show that the LP used for obtaining \Cref{thm:deterministic} has an integrality gap of $\ell$, and thus would need to be strengthened to achieve this end (see \Cref{sec:integralityGap}).  Hence, obtaining a $\textnormal{polylog} (\ell)$ competitive factor would require a new set of techniques. As we have already discussed earlier, existing techniques for obtaining randomized online algorithms for paging problems and related variants focus on solving covering linear programs (LP) online, and they break down in the presence of general non-linear paging constraints. 

To overcome the integrality gap of our LP, we formulate a stronger LP for non-linear paging (see \Cref{sec:integralityGap}). Obtaining even a fractional $\textnormal{polylog}(\ell)$-competitive algorithm for this LP seems a hard task. Thus, we consider another parameter which allows us to get a better competitive ratio. The parameter is the maximum number of pages that fit together in the cache. Formally, define 
\begin{equation}
	\label{eq:mu}
	\mu = \max_{S \subseteq \cP \text{ s.t. } f(S) \leq k} |S|
\end{equation} as the maximum cardinality of a feasible set in cache. Observe that $\mu = k$ for e.g., generalized paging, hence it is a natural parameter to also consider in our setting. We also remark that in many practical settings, such as classic paging, it holds that $n \gg \mu$ and obtaining a competitive ratio that depends on $\mu$ (rather than $n$) is much more desirable. 

Clearly, $\mu \geq \ell$. The following example illustrates a scenario demonstrating a scenario where $\mu \gg \ell$. Consider a non-linear paging
instance on a set P of pages such that P is partitioned into sets $X,Y$, where $|X| = n$, $|Y| = k+1$, and $n \gg k$.
Define a feasibility function $f$ such that for all subsets $S$ of $P$ define $f(S) = |Y \cap S|$. Define the cache
threshold as $k$. Therefore, the only minimally infeasible set is $Y$; thus, $\ell = k$. On the other hand, the
maximum cardinality of a set that fits into cache is the cardinality of all pages in $X$ and any $k$ pages from
$Y$; thus, $\mu = n+k$. Since $n \gg k$ (i.e., $n$ can be chosen to be arbitrarily large w.r.t. $k$) it follows that $\mu \gg \ell$. 

We give the following fractional algorithm for solving the strengthened LP online (see \Cref{sec:stronger}). %that this natural parameter %accurately describes the hardness of non-linear paging. %(rather than $k$). 

\begin{theorem}
	\label{thm:MU}
	There is an $O\left(\log \mu\right)$-competitive algorithm for obtaining a fractional solution for the strengthened \textnormal{LP}. %\textnormal{\textsf{Stronger-LP}} \eqref{eq:LPS}. 
\end{theorem}

It as an interesting open question if a randomized $\textnormal{polylog}(\mu)$-competitive algorithm for non-linear paging can be designed by rounding the fractional solution obtained in \Cref{thm:MU}.

\subsubsection{Supermodular Paging} % (Submodular Cover) }

We now discuss our results and techniques for supermodular paging. Our main result is a polylogarithmic randomized competitive algorithm for supermodular paging, i.e., submodular cover paging (see \Cref{sec:randomizedALG}). % in particular showing that online algorithm exist for large family of non-trivial non-linear paging instances without a dependence on $T$. 
As we show later on, our upper bound turns out to be quite close to the lower bound we prove. 

%Since the supermodular paging problem is a closely related variant of online set cover \cite{alon2003online} and online submodular cover \cite{gupta2020online}, a significant improvement of our results is unlikely to exist \cite{korman2004use}.  

	\begin{theorem}
	\label{thm:randomized}
There is an $O\left(\log^2 \mu \cdot \log \left(\frac{c_{\max}}{c_{\min}} \cdot f(\cP) \right)\right)$-competitive randomized algorithm for \textnormal{supermodular paging (submodular cover paging)}, where $c_{\max},c_{\min}$ are the maximum and minimum costs of pages, respectively.
\end{theorem}
In particular, for unweighted supermodular paging (where $c_p = 1~\forall p \in \cP$), the above theorem implies an $O\left(\log^2 (\mu) \cdot \log \left(f(\cP) \right)\right)$-competitive algorithm.

Our algorithm relies on a different LP relaxation than the LP described in \Cref{sec:ell}. As we aim to solve supermodular paging, which implies a submodular cover function $g$, we design an LP relaxation inspired by the submodular cover relaxation of Wolsey \cite{wolsey1982analysis} (see also \cite{gupta2020online}). We solve this LP where the constraints arrive online to obtain a fractional $O(\log(\mu))$-competitive (deterministic) algorithm, while maintaining the property that the entries are either integral or multiples of $\frac{1}{k}$. 

To obtain an integral solution, one is tempted to maintain online a set of feasible {\em cache states} respecting (even approximately) the marginal probabilities induced by the fractional solution, similarly to previous works on weighted paging and generalized paging \cite{bansal2012primal,bansal2012randomized,adamaszek2018log}. However, techniques for online cache state maintenance of \cite{bansal2012primal,bansal2012randomized,adamaszek2018log} seem to break down in the presence of submodular cover constraints. Instead, we augment the fractional solution by an additional polylogarithmic {\em boosting factor} and perform randomized rounding, with possible corrections to ensure feasibility. The probability of a page to be evicted at some point in time is shown to be proportional to the page's fractional increase normalized by the probability that the page is in cache. 

We also give lower bounds for online supermodular paging. Surprisingly, even though online set cover seems starkly different from supermodular paging, in particular since the cover constraints change over time for online set cover, we can show that supermodular paging is in a sense harder to solve online (see \Cref{sec:hardness}).

\begin{theorem}
	\label{lem:Set Cover}
	For any $\rho \geq 1$, if there is a $\rho$-competitive algorithm for \textnormal{supermodular paging}, then there is a $\rho$-competitive algorithm for \textnormal{online set cover} having the same running time up to polynomial factors.  
\end{theorem}

% we have the statement of \Cref{thm:hardness}. 

By \Cref{lem:Set Cover} and the results of \cite{alon2003online,korman2004use}, we give a lower bound on the competitiveness of supermodular paging, indicating the necessity of the factor $\log \mu \cdot \log \left(\cP\right)$.
%and $\log(f(\cP))$ 
%in the competitive ratio of \Cref{thm:randomized}. %The proof is based on a reduction from online set cover. 

\begin{cor}
	\label{thm:hardness}
	Unless \textnormal{NP $\subseteq $ BPP}, there is no polynomial $o\left(\log \mu \cdot \log \left(f(\cP)\right) \right)$-competitive algorithm for \textnormal{supermodular paging}. Moreover, there is no deterministic $o \left(\frac{\log \mu \cdot \log \left(f(\cP)\right)}{\log \log \mu+\log \log \left(f(\cP)\right)}\right)$-competitive algorithm for \textnormal{supermodular paging} of any running time. 
\end{cor}

We remark that some of our results from \Cref{sec:ell} can be stated using the parameter $\mu$ as well. However, as we showed earlier, there are non-linear paging instances in which the parameter $\mu \gg \ell$ and it grows artificially apart
from the true hardness of the instance – in terms of competitive analysis. In contrast, every minimally
infeasible set of cardinality $\ell+1$ can be used to obtain the lower bounds for classic paging \cite{sleator1985amortized,fiat1991competitive} (i.e., $\ell$-
deterministic and $O(\log \ell)$-randomized lower bounds). Thus, in an informal sense, an increase in $\ell$
always incurs an increase in the difficulty of the problem, unlike an increase in $\mu$. For this reason, we
believe that searching for competitive algorithms and lower bounds in terms of $\ell$, rather than $\mu$, may be of
greater interest in the non-linear paging setting.

Finally, we also aim at obtaining an offline approximation algorithms for supermodular paging (i.e., where all requests are known in advance). A first attempt would be to reduce the problem to offline submodular cover, as follows. Define a covering function $G$ on the domain containing all pairs $(p,j)$, for every page $p$ and the $j$-th time it is requested. Define the value of $G$ on a subset of pairs $S$ as $G(S) = \sum_{t \in T} g(S_t)$, where $g$ is the covering function of the instance (see \Cref{sec:ell}), and $S_t$ is the set of all pages $p$, where $(p,j)$ belongs to $S$ and the time interval between requests $j$ and $(j+1)$ for $p$ intersects $t$. Clearly,  the total cover demand, even with a unit demand per-time slot, is a function of $T$; thus, applying an offline set cover algorithm in a black box manner gives only an $\Omega\left(\log \left(T\right)\right)$-approximation \cite{lund1994hardness}. As $T$ may be very large, a different technique is in place.
%(see \Cref{sec:technicalC} for more details on covering ).

Combined with the strong {\em round-or-separate} algorithm of \cite{gupta2020online}, our techniques yield an approximation algorithm for supermodular paging independently of $T$ (see \Cref{sec:randomizedALG}).
We defer the details to the full version of the paper.
\begin{theorem}
	\label{thm:approximation}
	There is an offline $O\left( \log \left(\frac{c_{\max}}{c_{\min}} \cdot f(\cP) \right)\right)$-approximation algorithm for \textnormal{supermodular paging}. %(submodular cover paging)}%, where $c_{\max},c_{\min}$ are the maximum and minimum costs in the input. 
	\end{theorem}

\section{Deterministic Algorithm for Non-Linear Paging}
\label{sec:algUneweighted}

%In this section, we derive a tight deterministic $\ell$-competitive algorithm for general non-linear paging. 
In this section, we present a deterministic algorithm for non-linear paging and show that it gives a tight competitive ratio for the problem, thus providing proofs for \Cref{thm:deterministic} and \Cref{thm:generalized}. Our algorithmic results are based on a new LP relaxation for the problem which we introduce here.
%s apply the online primal-dual method together with a new LP which is described below.  

\subsection{LP Relaxation for Non-Linear Paging}
\label{sec:LP}

   Consider a non-linear paging instance with a set of pages $\cP$, a feasibility function $f$, cache threshold $k$, a set of time points $T$, and a page request $p_t \in \cP$ for every time $t \in T$. To simplify notation, for a set $S$ and an element $p$ we use $S+p = S \cup \{p\}$ and $S-p = S \setminus \{p\}$. The LP is as follows.
 %We use the following {\em linear programming (LP)} relaxation of non-linear paging. %For the following, fix a sequence of requests $I \subseteq \om^{\mathbb{N}}$. For some $w \in \mathbb{N}$, let $[w] = \{1,2,\ldots,w\}$. For every $p \in \om$ and $t \in T(I)$ let $R(p,t) = \left\{ t' \in [t]~|~I_{t'} = p \right\}$ be the set of time points until time $t$ in which the requested page is $p$; in addition, let $r(p,t) = \max_{t' \in R(p,t)} t'$ be the time $j$ of the last request for $p$ up to time $t$; assume that $r(p,t) = -1$ if $R(p,t) = \emptyset$. We also use $R(p,j) = \{t \in T(I)~|~j = r(p,t)\}$ for all time points between the $j$-th request to $p$ to the last time point before the $(j+1)$-th request for page $p$. 
 
The variables of the LP are $x_{p}(j)$ for every page $p \in \cP$ and the $j$-th time that $p$ is requested.  Intuitively, the value of the variable $x_{p}(j)$ indicates to what extent page $p$ is evicted between its $j$-th and $(j+1)$-th requests.  %For some page $t \in T(I)$, let $\bar{\cC}(I_t) = \{S \subseteq \om~|~|S|>k, I_t \in S\}$ be the set of all non-valid cache states of $I_t$ that contain $I_t$. These sets are called {\em violating sets}.    
The constraints state that for every infeasible set $S$ containing the requested page $p_t$, for time point $t$, we must ``break" this set - i.e., evicting at least one page from $S - p_t$. If this constraint is satisfied for every such infeasible set by an integral solution, then it induces a feasible set of pages in the cache at all times. See \Cref{fig:M1} for a visualization of these constraints. %The LP is as follows. 

We use $n_p$ to denote the number of requests for page $p$ during the request sequence. %and let $[n_p] = \left[n_p\right]$; 
Also, %with a slight abuse of notation, 
we use $r(p,t)$ to denote the number of requests for page $p$ until time $t$. We assume that $x_p(0)$ is a variable always set to $0$, for all $p \in \cP$. To simplify notation, for every $t \in T$, let $\cS(t) = \left\{ S \subseteq \cP ~|~  p_t \in S \textnormal{ and } f(S)>k \right\}$ denote the infeasible sets containing $p_t$. For any $w \in \mathbb{N}$, let $[w] = \{1,2,\ldots,w\}$. Our LP relaxation is as follows. 

%Finally, let $T$ be the set if all time points. %(including $t$). %variable $x_p(j)$ where there are exactly   %require that the total fractional amount taken from each violating set $S \in \bbc$ is at least $|S|-k$; this is the minimum number of pages  from $S$ that are required to be outside of the cache in any valid cache state of $I_t$. %For the corner case where $x_p(r(p,t)) = x_p(-1)$, i.e., for $R(p,t) = \emptyset$, assume that $x_p(-1) = 1$. This ensures the feasibility of all constraints and does not incur an additional cost for the LP. The LP is defined as follows. 
	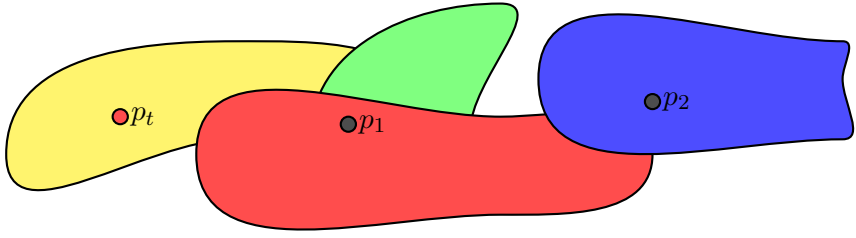
\begin{figure}
	\centering
	
	\begin{tikzpicture}[thick]
		%	\tikzset{vertex/.style = {shape=circle,draw,minimum size=1.5em}}
		%\tikzset{edge/.style = {->,> = latex'}}
		%\tikzstyle{edge} = [fill,opacity=.5,fill opacity=.5,line cap=round, line join=round, line width=50pt]
		\tikzstyle{edge} = [->, line width=1pt]
		%\tikzset{edge/.style = {->,> = latex'}}
		\node (v1) at (0,2) {};
		\node (v5) at (9,5) {};
		\node (v6) at (9,4.5) {};
		\node (v7) at (9,4) {};
		
		%									\node (v8) at (9,3.5) {};
		%									\node (v9) at (9,3) {};
		%									\node (v10) at (9,2.5) {};
		%									
		%									\node (v11) at (9,2) {};
		%									\node (v12) at (9,1.5) {};
		%									\node (v13) at (9,1) {};
		%									
		%									\node (v14) at (9,0.5) {};
		%									\node (v15) at (9,0) {};
		%									\node (v16) at (9,-0.5) {};
		%									
		%									\node (v17) at (13,3.5) {};
		%									\node (v18) at (13,1.5) {};
		
		\node (v2) at (1.5,3) {};
		%	\node at (2,1) {$\bf \textcolor{black}{U}$};
		\node (v3) at (4,2.5) {};
		\node (v4) at (5,2) {};
		\node (v5) at (7,3) {};
		%	\node at (7,1) {$\bf \textcolor{black}{V}$};
		%\node at (12,1) {$\bf \textcolor{black}{W}$};
		\node (u0) at (4,2) {};
		\node (u1) at (5,2) {};
		\node (u2) at (6,2) {};
		
		\node (x0) at (8.5,3) {};
		\node (x1) at (9.5,3) {};
		\node (x2) at (10.5,3) {};

		%	\node (p2) at (1,2.5) {$\bf \textcolor{red}{p_t}$};
		
		\begin{scope}[fill opacity=0.4]
			\filldraw[fill=yellow!70] ($(v1)+(-0.5,0)$) 
			to[out=90,in=180] ($(v2) + (1,0.5)$) 
			to[out=0,in=90] ($(v3) + (1,0)$)
			to[out=270,in=0] ($(v2) + (1,-0.8)$)
			to[out=180,in=270] ($(v1)+(-0.5,0)$);

			\filldraw[fill=green!50] ($(v4)+(-1.5,0.2)$)
			to[out=90,in=180] ($(v4)+(1,2)$)
			to[out=0,in=90] ($(v4)+(0.6,0.3)$)
			to[out=270,in=0] ($(v4)+(1,-0.6)$)
			to[out=180,in=270] ($(v4)+(-1.5,0.2)$);

			%	\begin{scope}
				%	\filldraw[fill opacity=0.9]
				\filldraw[fill=red!70] ($(u1)+(-3,0)$) 
				to[out=90,in=180] ($(u2) + (0,0.5)$) 
				to[out=0,in=90] ($(u0) + (4,0)$)
				to[out=270,in=0] ($(u2) + (0,-0.8)$)
				to[out=180,in=270] ($(u1)+(-3,0)$);
				
				\filldraw[fill=blue!70] ($(x1)+(-3,0)$) 
				to[out=90,in=180] ($(x2) + (0,0.5)$) 
				to[out=0,in=90] ($(x0) + (2,0)$)
				to[out=270,in=0] ($(x2) + (0,-0.8)$)
				to[out=180,in=270] ($(x1)+(-3,0)$);
			\end{scope}

			%	\filldraw[fill=brown!70] (v1) circle (0.1) node [right] {$u_1$};
			%\filldraw[fill=darkgray!70] (v2) circle (0.1) node [below left] {$u_2$};
			%\filldraw[fill=blue!70] (v3) circle (0.1) node [left] {$u_3$};
			%\filldraw[fill=red!70] (v4) circle (0.1) node [below] {$v_1$};
			%\filldraw[fill=purple!70] (v5) circle (0.1) node [below] {$v_2$};
			%\filldraw[fill=white!70] (u1) circle (0.1) node [right] {$w_3$};
			\filldraw[fill=black!70] (4,2.4) circle (0.1) node [right] {$p_1$};
			
			\filldraw[fill=black!70] (8,2.7) circle (0.1) node [right] {$p_2$};
			
			\filldraw[fill=red!70] (1,2.5) circle (0.1) node [right] {$p_t$};
			%	\filldraw[fill=pink!70] at (4,2.5) circle (0.1) node [left] {$p_1$};
		\end{tikzpicture}
		\caption{\label{fig:M1} An illustration of the LP constraints. Each colored shape represents a minimally infeasible set of pages upon the arrival of a page $p_t$. Removing pages $p_1$ and $p_2$ ensures cache feasibility.}  
	\end{figure}

\begin{equation}
	\label{eq:LP}
	\begin{aligned}
		%\textsf{Primal-LP}:
  ~~~~~~~~	& \min~~~ \sum_{p \in \cP~} \sum_{j \in [n_p]}  x_p(j) \cdot c(p)\\
		& ~~~~~~~~\text{s.t. }\\
		& \sum_{p \in S-p_t} x_p(r(p,t)) \geq 1, 	~~~~~~~~~~~~~~\forall t \in T~\forall S  \in \cS(t) \\ 
		& x_p(j) \geq 0 ~~~~~~~~~~~~~~~~~~~~~~~~~~~~~~~	\forall p \in \cP~\forall j \in [n_p]
	\end{aligned}
\end{equation} 
We now define the dual LP. For page $p \in \cP$ and $j \in [n_p]$ let $I(p,j) = \{t \in T~|~j = r(p,t)\}$ be the {\em interval} of all time points between the $j$-th request to $p$ till the last time point before the $(j+1)$-st request for page $p$. %$I(p,j)$ be the interval of time points from the $j-1$-th request to $p$ to the $j$-th request to $p$.   %For short, given $p \in \om$, variables $y_t(S)$ for $t \in T(I)$ and $S \in \bbc$, let $$y_t(p) = \sum_{S \in \bbc~|~p \in S-I_t} y_t(S)$$
%For every $p \in \om$ and $j \in R(p,T(I))$ let $R(p,j) = \left|\left\{ t' \in [t]~|~I_{t'} = p \right\}\right|$. 
The {\em dual} of \eqref{eq:LP} is the following. 
\begin{equation}
	\label{eq:dualD}
	\begin{aligned}
		%\textsf{Dual-LP}:
  ~~~~~~~~	& \max~~~ \sum_{t \in T~} \sum_{S  \in \cS(t)} y_t(S)\\
		& ~~~~~~~~\text{s.t. }\\
		& \sum_{t \in I(p,j)~} \sum_{S \in \cS(t)\big|p \in S-p_t} y_t(S) \leq c(p), 	~~~~~~~~~~~~~~\forall p \in \cP~\forall j \in [n_p]\\ 
	%	& y_t(S) \geq 0 ~~~~~~~~~~~~~~~~~~~~~~~~~~~~~	\forall t \in T(I) ~\forall S \in \bbc. 
	\end{aligned}
\end{equation}

%To further simplify the notations, for $p \in \om$ and $j \in R(p,T(I))$ let $y_p(j) = \sum_{t \in R(p,j)} y_t(p)$ (i.e., the left side of the dual constraints). 
We use $\textsf{Primal-LP}$ and $\textsf{Dual-LP}$ to denote the values of the optimal solutions for the primal and dual programs, respectively, and by $\OPT$ the offline (integral) optimum.  
%The LP \eqref{eq:LP} is used in \cite{bansal2012randomized} in a generalization of weighted paging, and it is in particular a relaxation of a solution for weighted paging. Thus, 
%Using weak duality we have the following result. 
The next result follows from weak duality and since our LP is a relaxation of non-linear paging.
  
\begin{lemma}
	\label{lem:relaxation2}
%	For every unweighted \textnormal{non-linear paging} instance $I$ it holds that 
$\textnormal{\textsf{Dual-LP}} \leq \textnormal{\textsf{Primal-LP}} \leq \OPT$.
\end{lemma}

\begin{proof}
	The first inequality follows from weak duality, since \eqref{eq:dualD} is the dual program of \eqref{eq:LP}. For the second inequality, we show that \eqref{eq:LP} is a relaxation of non-linear paging. Consider an integral feasible solution $M$ to our instance; define a solution $x$ for \eqref{eq:LP} such that $x_p(j) = 1$ if and only if page $p$ is removed from the cache during the interval $I(p,j)$ in $M$, for every $p \in \cP$ and $j \in [n_p]$.  Then, for every time $t \in T$ and $S \in \cS(t)$, $S$ cannot be fully contained in the cache at time $t$, since $M$ is a feasible solution, and there is at least one $p \in S-p_t$ which is not in the cache at this time. Therefore, the primal constraint corresponding to $t$ and $S$ is satisfied. We conclude that $x$ is feasible for \eqref{eq:LP}; consequently, \eqref{eq:LP} is indeed a relaxation of non-linear paging, implying the second inequality.   
\end{proof} 

%By the proof of \Cref{lem:relaxation2}, we can use interchangeably the terms {\em solution} for the instance and an integral primal solution for \eqref{eq:LP}.  

%\sub
\subsection{A Deterministic Algorithm for Non-Linear Paging}
%\label{sec:algUneweighted}

%In this section, we present a deterministic algorithm for non-linear paging, and show that it gives a tight competitive ratio for the problem; this gives the proof of \Cref{thm:deterministic} and \Cref{thm:generalized}. 

In this section, we give a primal-dual algorithm, based on the LP relaxation of non-linear paging presented in \Cref{sec:LP}. For brevity, we denote the left handside of the dual constraint corresponding to page $p$ and $j \in [n_p]$ as
\begin{equation}
	\label{eq:Y}
	Y_p(j) = \sum_{t \in I(p,j)~} \sum_{S \in \cS(t)\big|p \in S-p_t} y_t(S).
\end{equation}
 We call a page $p$ {\em tight} at time $t$ if $Y_p(j) = c(p)$.  

The algorithm initializes an infeasible primal solution and a feasible dual solution as vectors of zeros $\bar{0}$. In every time step $t$, if the set of pages currently in cache, denoted by $\textsf{Cache}_t$, is infeasible, then our algorithm finds a subset of pages $Q$ in the cache such that: first, $Q$ is infeasible; second, $Q$ contains the requested page $p_t$, i.e., $Q \in \cS(t)$; third, $Q$ has minimum cardinality amongst all such sets.   %We remove all pages in $Q$ except $p_t$ from cache and 
we increase the variable $y_{t}(Q)$ continuously. %and simultaneously increase at the same rate the variables $x_p(r(p,t))$ for all pages in $Q$ except $p_t$. 
Once one of the pages becomes tight, we remove it from cache. If the cache is feasible, the algorithm proceeds to time $t+1$; otherwise, we repeat this process with a new set $Q'$ from the current cache.  The pseudocode is given in \Cref{alg:deterministic}.

\begin{algorithm}[h]
	\caption{$\textsf{Deterministic}$}
	\label{alg:deterministic}

%	Initialize an empty cache $S \leftarrow \emptyset$ and an empty buffer $B \leftarrow \emptyset$.
	
%	Initialize the cache $\cG_0 \leftarrow \emptyset$.
%	
Initialize (infeasible) primal and (feasible) dual solutions $x \leftarrow \bar{0}$ and   $y \leftarrow \bar{0}$. 

	\For{time $t \in T$}{
		
		%Add $p_t$ to cache  
	%	initialize $x_p(r(p,t)) \leftarrow 0$.
		
		Let $\textsf{Cache}_t = \left\{p \in \cP~|~r(p,t) \geq 1 \textnormal{ and } x_p(r(p,t)) = 0\right\}$ be the pages currently in cache.  
		
	%	$S \leftarrow \cC+p_t$ and remove from buffer $B \leftarrow B-p_t$. 
		
		\While{$f\left(\textnormal{\textsf{Cache}}_t\right)>k$\label{step:while}}{

%		\If{$B \neq \emptyset$}{
%		
%		Remove some page $p \in B$ from the cache $S \leftarrow S-p$ and from buffer $B \leftarrow B-p$. 
%		
%		}
%		
%		\Else{
		
		Find  $Q \subseteq \textsf{Cache}_t$ such that $Q \in \cS(t)$ of minimum cardinality.\label{step:Q}
		
	%	Update the buffer $B \leftarrow Q$.

		%find $q \in I(p,r(p,t))$ such that $y_{q}(S) = 0 ~\forall S \in \cS(q)$ and $p_{q} \in Q$
		
		\While{$Y_p(r(p,t)) < c(p)~\forall p \in Q-p_t$\label{step:Y<c}}{
		
		Increase $y_{t}(Q)$ %and $x_{p}(r(p,t))~\forall p \in Q-p_t$ 
		continously.\label{step:t'}
		
		}

		\For{$p \in Q-p_t$ such that $Y_p(r(p,t)) = c(p)$}{
		
			Remove $p$ from cache: $x_p(r(p,t)) \leftarrow 1, ~~\textsf{Cache}_t \leftarrow \textsf{Cache}_t-p$.\label{step:remov}
		}
				%For every : 

	%	 Increase $y_{t}(Q)$ and $x_{p'}(r(p,t))~\forall p' \in Q-p_t$ continously until the primal constraint of some $p \in Q-p_t$ becomes tight.\label{step:t'}
		
	%	}

		}

	}

%Set all undefined dual variables to zero. 

Return the (primal) solution $x$ and the (dual) solution $y$. 
\end{algorithm}

We start by showing the feasibility of the algorithm. 

	\begin{lemma}
	\label{lem:feasibilityD}
	\Cref{alg:deterministic} returns primal and dual feasible solutions of \eqref{eq:LP}.  
\end{lemma}

\begin{proof}
	%The primal is feasible 
	By \Cref{step:while,step:remov}, the algorithm keeps evicting pages until reaching a feasible set of pages in the cache. By the monotonicity of the feasibility function $f$, we eventually reach a feasible set in cache: every page $p$ is considered to be feasible alone in the cache, i.e., $f(\{p\}) \leq k$; thus, in the worst case, the content of the cache at the end of time step $t$ is $p_t$. The above shows that the primal solution $x$ is feasible, being a relaxation of the integer problem.
	In addition, the dual $y$ is feasible since for all $p \in \cP$ and $j \in [n_p]$ it holds that $Y_p(j) \leq c(p)$ by \Cref{step:Y<c}. 
\end{proof}

In the following we bound the competitive ratio of the algorithm. %Using \Cref{lem:feasibilityD} and 
Let $c(x)$ be the cost of our (integral) solution $x$ and let $v(y)$ be the value of the dual solution $y$ obtained by \Cref{alg:deterministic}. Using the selection of minimal sets for $Q$ in \Cref{step:Q} we have the following result. 
Recall the width parameter $\ell$ defined in \Cref{def:Mininfeasible}.

	\begin{lemma}
	\label{lem:k}
$c(x) \leq \ell \cdot v(y)$
\end{lemma}

\begin{proof}
By LP \eqref{eq:LP},
\begin{equation}
	\label{eq:k1}
	\begin{aligned}
	%c(x) ={} & \sum_{p \in \cP~} \sum_{j \in [n_p]}  x_p(j) \cdot c(p)\\
	%\leq{} &  \sum_{p \in \cP~} \sum_{j \in [n_p]} x_p(j) \cdot Y_p(j)\\
 %={} &   \sum_{p \in \cP~} \sum_{j \in [n_p]} x_p(j) \cdot \left(  \sum_{t \in I(p,j)~} \sum_{S \in \cS(t)\big|p \in S-p_t} y_t(S) \right)\\
 c(x) ={} & \sum_{p \in \cP~} \sum_{j \in [n_p]}  x_p(j) \cdot c(p)
	\leq{}   \sum_{p \in \cP~} \sum_{j \in [n_p]} x_p(j) \cdot Y_p(j) \\
	={} &   \sum_{p \in \cP~} \sum_{j \in [n_p]} x_p(j) \cdot \left(  \sum_{t \in I(p,j)~} \sum_{S \in \cS(t)\big|p \in S-p_t} y_t(S) \right).\\
	\end{aligned}
\end{equation} The inequality holds since we only evict tight pages. By changing summation order in \eqref{eq:k1}, 
\begin{equation*}
	\label{eq:k2}
	\begin{aligned}
		c(x) \leq{} & \sum_{t \in T~} \sum_{S  \in \cS(t)} y_t(S) \cdot \sum_{p \in S -p_t} x_p(r(p,t)) \leq \sum_{t \in T~} \sum_{S  \in \cS(t)} y_t(S) \cdot \ell = \ell \cdot v(y). 
	\end{aligned}
\end{equation*} The second inequality holds since in \Cref{step:Q} we always choose a minimally infeasible set of cardinality at most $\ell +1$; hence, $y_t(S) \neq 0$ only if $|S| \leq \ell+1$. 
\end{proof}
%Using \Cref{lem:k}, we can give the proof of the main theorem. 
\noindent{}
We are now ready to give the proof of our main theorem.

%\subsubsection*{Proof of \Cref{thm:deterministic}}
\paragraph{Proof of \Cref{thm:deterministic}:}
By \Cref{lem:feasibilityD}, the returned solution $x$ is a feasible solution. Combined with \Cref{lem:relaxation2} and \Cref{lem:k}, %for an instance $I$ with width $\ell$ we have %the statement of the lemma. 
$c(x) \leq \ell \cdot v(y) \leq \ell \cdot \OPT$.
%By \Cref{lem:k} is of cost at most $\ell(f^*) \cdot \OPT^*$. Therefore, by \Cref{lem:widthf^*} and \Cref{lem:OPT} the cost paid by the returned solution is at most $\ell(f) \cdot \OPT$. 
%a deterministic $\ell(f^*)$-competitive algorithm for non-linear paging. 
Along with the tight lower bound from the special case of paging \cite{sleator1985amortized}, the statement of the theorem follows. 

We remark that for any $\ell$ and function $f$ such that $\ell(f) = \ell$, there is a minimally infeasible set $S$ of cardinality $\ell$; as a result, the lower bound for deterministic algorithms of paging \cite{sleator1985amortized} can be obtained from the set $S$. Namely, given a deterministic algorithm $\cA$ for unweighted non-linear paging, construct a sequence of requests for pages in $S$ that each time asks for the page missing in cache by algorithm $\cA$; clearly, $\cA$ pays a cost of one at each time step since $S$ is minimally infeasible. Conversely, the offline optimum would evict the page that will be requested latest in the future, missing only once in $\ell$ requests, giving the lower bound of $\ell$. 
\qed

\subsection{Integrality Gap}
\label{sec:integralityGap}
In this section, we show the limitations of the LP. Specifically, we show that the integrality gap of the LP defined in \eqref{eq:LP} is at least $\ell$; together with our algorithmic upper bound the integrality gap is exactly $\ell$. We then discuss a stronger LP formulation based on \eqref{eq:LP}. 

\begin{lemma}
    \label{lem:integralityGap}
    The integrality gap of \textnormal{LP} \eqref{eq:LP} is $\ell$. 
\end{lemma}
\begin{proof}
For some $n \geq 1$, consider an instance with $n$ pages $\cP = \{p_1,\ldots, p_n\}$ with uniform costs $c(p) = 1~\forall p \in \cP$ and a uniform feasibility function $f$ as in classic paging, i.e.,  $f(S) = |S|$ for all $S \subseteq \cP$, and some cache capacity $k$. Thus, all minimally infeasible sets are of cardinality $k+1$ and it follows that $\ell(f) = k$. Define the sequence of requests $p_1,\ldots, p_n$. That is, each page is requested exactly once. Observe that each request results in a page fault. Hence, any integral solution, in particular the optimal integral solution, evicts at least $n-k$ pages. On the other hand, define a fractional solution such that $x_p(1) = \frac{1}{k}$ for all $p \in \cP$, i.e., each page is evicted after its first request with fraction $\frac{1}{k}$ (note that each page is requested exactly once). Consider some infeasible set $S \subseteq \cP$ such that $|S| > k$. It holds that 
$$\sum_{p \in S} x_p(1) = |S| \cdot \frac{1}{k} \geq \frac{k+1}{k} \geq 1.$$
Thus, $x$ satisfies the constraints of \eqref{eq:LP}. The cost paid by the fractional solution $x$ is $\frac{n}{k}$. Therefore, as $n$ and $k$ can be chosen arbitrarily, the integrality gap of the LP is at least $$\lim_{n \rightarrow \infty} \frac{n-k}{\frac{n}{k}} = \lim_{n \rightarrow \infty} \left(k-\frac{k^2}{n}\right) = k = \ell(f).$$
    Therefore, in general, the integrality gap cannot be smaller than $\ell$.  
\end{proof}

\subsubsection{Discussion: a Stronger LP}
The integrality gap example shows that LP \eqref{eq:LP} is not sufficient for obtaining a randomized $\textnormal{polylog}(\ell)$-competitive algorithm for general non-linear paging. Instead, we describe a stronger version of our LP \eqref{eq:LP}, %
in which we require removing from each infeasible set $S$ a set of pages $S'$, so that the complement of $S'$ in $S$, $S \setminus S'$, will be feasible in the cache (i.e., $ f\left(S \setminus S'\right) \leq k$). The strengthened LP and our fractional algorithm for solving the LP online are presented in \Cref{sec:stronger}, giving the proof of \Cref{thm:MU}.

\section{A Randomized Algorithm for Supermodular Paging} 
\label{sec:randomizedALG}

We provide here an $O(\log \mu)$-competitive algorithm for a fractional version of supermodular paging. Then, we design randomized online and offline algorithms for supermodular paging. %Due to space constraints, 
Some of the proofs are given in \Cref{sec:omit}. %The algorithm initializes the (infeasible) 

\subsection{LP for Supermodular Paging}

The starting point of our fractional algorithm is earlier work on submodular cover LP \cite{wolsey1982analysis,gupta2020online,coester2022competitive}. Consider a supermodular paging instance with a set of pages $\cP$, a feasibility function $f$ inducing the submodular cover function $g$ (recall that $g(S) = f(\cP)-f(\cP \setminus S)$ for every $S \subseteq \cP$), cache threshold $k$, a set of time points $T$, and a page request $p_t \in \cP$ for every time $t \in T$. We use the following LP relaxation of supermodular paging. As in the LP introduced in \Cref{sec:LP}, the variables of the LP are $x_{p}(j)$ for every page $p \in \cP$ and the $j$-th time that $p$ is requested. We will use the same notation as in the LP of \eqref{eq:LP}. Recall that for some $S \subseteq P$ and $p \in \cP$ we use $g_S(p) = g_S(\{p\}) = g(S+p)-g(S)$. Let $N = f(\cP)-k$ be our {\em cover demand}; at any point of time, the (non-linear) total size outside of the cache must be at least $N$. 

Our LP relaxation goes as follows; the constraints of the LP require that for every time $t$ and subset of pages $S$ (assumed to already be outside of the cache) we must evict from the cache (fractionally) a total size of at least $N-g(S) = f \left(\cP\setminus S\right)-k$, since we cannot have more than a total size of $k$ in cache. This constraint is very natural in the linear case (i.e., classic paging), but more involved in the submodular cover setting. 
   
%The constraints of the LP require that for every infeasible set $S$ that contains the requested page $p_t$ of some time point $t$, we must "break" this set - evicting at least one page from $S - p_t$; as we do so for every such  infeasible set, an integral solution induces a feasible set of pages in the cache at all times. 
%We use $n_p$ to denote the number of requests for a page $p$ during the request sequence; %and let $[n_p] = \left[n_p\right]$; 
%also, %with a slight abuse of notation, 
%we use $r(p,t)$ to denote the number of requests for page $p$ until time $t$. We also use the technical assumption that $x_p(0)$ is a variable always set to $0$, for all $p \in \cP$. To simplify the notation, for every $t \in T$ let $$\cS(t) = \left\{ S \subseteq \cP ~|~  p_t \in S \textnormal{ and } f(S)>k \right\}$$

%be the set of infeasible sets containing $p_t$. Our LP relaxation is as follows. 

%Finally, let $T$ be the set if all time points. %(including $t$). %variable $x_p(j)$ where there are exactly   %require that the total fractional amount taken from each violating set $S \in \bbc$ is at least $|S|-k$; this is the minimum number of pages  from $S$ that are required to be outside of the cache in any valid cache state of $I_t$. %For the corner case where $x_p(r(p,t)) = x_p(-1)$, i.e., for $R(p,t) = \emptyset$, assume that $x_p(-1) = 1$. This ensures the feasibility of all constraints and does not incur an additional cost for the LP. The LP is defined as follows. 

\begin{equation}
	\label{eq:LP2}
	\begin{aligned}
%		\textsf{Primal-LP}:
~~~~~~~~	& \min \sum_{p \in \cP~} \sum_{j \in [n_p]}  x_p(j) \cdot c(p)\\
		& ~~~~~~~~\text{s.t. }\\
		& \sum_{p \in \cP-p_t} x_p(r(p,t)) \cdot g_S(p) \geq N-g(S), 	~~~~~~~~~\forall t \in T~\forall S  \subseteq \cP \\ 
		& x_p(j) \geq 0 ~~~~~~~~~~~~~~~~~~~~~~~~~~~~~~~~~~~~~~~~~~~~~~	\forall p \in \cP~\forall j \in [n_p]
	\end{aligned}
\end{equation} 

In the following we define the {\em dual} LP of \eqref{eq:LP2}. 

%For a page $p \in \cP$ and $j \in [n_p]$ let $$I(p,j) = \{t \in T~|~j = r(p,t)\}$$ be the {\em interval} of all time points between the $j$-th request to $p$ to the last time point before the $(j+1)$-th request for page $p$. %$I(p,j)$ be the interval of time points from the $j-1$-th request to $p$ to the $j$-th request to $p$.   %For short, given $p \in \om$, variables $y_t(S)$ for $t \in T(I)$ and $S \in \bbc$, let $$y_t(p) = \sum_{S \in \bbc~|~p \in S-I_t} y_t(S)$$
%%For every $p \in \om$ and $j \in R(p,T(I))$ let $R(p,j) = \left|\left\{ t' \in [t]~|~I_{t'} = p \right\}\right|$. 
%The {\em dual} of \eqref{eq:LP} is the following. 
\begin{equation}
	\label{eq:dual}
	\begin{aligned}
	%	\textsf{Dual-LP}:
 ~~~~~~~~	& \max \sum_{t \in T~} \sum_{S  \subseteq \cP} y_t(S) \cdot \left( N-g(S) \right)\\
		& ~~~~~~~~\text{s.t. }\\
		& \sum_{t \in I(p,j)~} \sum_{S \subseteq \cP \big|p \in S-p_t} y_t(S) \leq c(p), 	~~~~~~~\forall p \in \cP~\forall j \in [n_p]\\ 
			& y_t(S) \geq 0 ~~~~~~~~~~~~~~~~~~~~~~~~~~~~~~~~~~~~~~	\forall t \in T ~\forall S \subseteq \cP. 
	\end{aligned}
\end{equation}

%To further simplify the notations, for $p \in \om$ and $j \in R(p,T(I))$ let $y_p(j) = \sum_{t \in R(p,j)} y_t(p)$ (i.e., the left side of the dual constraints). 
We use $\textsf{Primal-LP}$ and $\textsf{Dual-LP}$ to denote the values of the optimal solutions for the primal and dual programs, respectively, and by $\OPT$ the offline (integral) optimum. 

Before we describe our fractional algorithm, we show that it is sufficient to satisfy only {\em minimal constraints} rather than all constraints of the LP. Formally,
a primal constraint of \eqref{eq:LP2} corresponding to $t \in T$ and $S \subseteq \cP$ is called {\em minimal} for some solution $x'$ if for all $p \in \cP$ where $x'_p(r(p,t)) = 1$ it holds that $p \in S$. %if $\cP \setminus S$ is a minimally infeasible set.\footnote{See \Cref{def:Mininfeasible} for a reminder of what minimally infeasible sets are.}

	\begin{lemma}
	\label{claim:z}
	If $x'$ satisfies all minimal constraints of \eqref{eq:LP2}, then $x'$ is feasible for \eqref{eq:LP2}.
\end{lemma}
Our fractional algorithm initializes the (infeasible) primal solution $x$ and the (feasible) dual solution $y$ both as vectors of zeros $\bar{0}$. When the requested page $p_t$ at time $t$ arrives, we do the following til $x$ satisfies all LP constraints  \eqref{eq:LP2} up to time~$t$. 

At time $t$, the algorithm considers a {\em minimally violating set} of pages $Q \subseteq \cP$. This set is a minimal set for our solution $x$ violating the primal constraint in \eqref{eq:LP2} corresponding to $t$ and $Q$. %and is of maximum cardinality amongst all such sets.   %We remove all pages in $Q$ except $p_t$ from cache and 
We increase variable $y_{t}(\cP \setminus Q)$ continuously and at the same time increase variables $x_p(r(p,t))$ for all pages in $(\cP-p_t) \setminus Q$ except $p_t$. The increasing rate is a function of $Y_p(r(p,t))$, where %recall that $Y_p(r(p,t))$
\begin{equation}
	\label{eq:Y2}
	Y_p(j) =\sum_{t \in I(p,j)~} \sum_{S \subseteq \cP \big|p \in S-p_t} y_t(S)
\end{equation}  is the left hand side of the corresponding dual constraint of $p$ and $j = r(p,t)$ in \eqref{eq:LP2} (analogously to \eqref{eq:Y}). This growth function has an exponential dependence on $Y_p(r(p,t))$ scaled by the cost of the page $c(p)$ and the number of pages possible in cache - the parameter $\mu$. %, and the total cover demand. 
The growth of the variable $x_p(j)$ stops once it reaches $\frac{1}{2}$.

%To compute $x_p(j)$, we continuously increase the (dual) variable $y_t(Q)$, corresponding to the non-valid set of pages in the cache at time $t$. At the same time, we increase the probability $x_p(t)$ (and equivalently, $X_p(t)$) according to an exponential function that depends on $y_p(j)$, where recall that $y_p(j)$ is the sum over all dual variables corresponding to $p$ in the interval $R(p,j)$. Once a valid distribution is constructed, we sample a page according to $x(t)$ and achieving a valid cache state once again. We also remove from the cache pages $p'$ for which $X_{p'}(t) \geq 1$. The pseudocode of the algorithm is given in \Cref{alg:fractional}. 

\begin{algorithm}[h]
	\caption{$\textsf{Fractional}$}
	\label{alg:fractional}
	
%	Initialize the cache $\cG_0 \leftarrow \emptyset$.
	Initialize (infeasible) primal solutions $x,z \leftarrow \bar{0}$ 
	
	Initialize (feasible) dual solution $y \leftarrow \bar{0}$.

	\For{$t \in T$}{
		
		\While{$x$ is not feasible for $t$\label{step:Whilefeasible}}{
		
		%	Let $\textsf{Frac}_t = \left\{p \in \cP~|~r(p,t) \geq 1 \textnormal{ and } x_p(r(p,t)) < 1\right\}$ be pages not fully evicted.   
		
		Find a minimal set $Q \subseteq \cP$ for $x$ %of maximum cardinality 
		such that $\sum_{p \in \cP-p_t} x_p(r(p,t)) \cdot g_Q(p) < N-g(Q)$.\label{step:fQ}

		\While{$\sum_{p \in \cP-p_t} x_p(r(p,t)) \cdot g_Q(p) < N-g(Q)$ \textnormal{ and $Q$ is minimal for $x$}\label{step:InWhile}}{
		
		Increase $y_t(\cP \setminus Q)$ continously.\label{step:fy}

			\ForAll{$p \in (\cP-p_t) \setminus Q$ }{
			
					increase $x_p(r(p,t))$ according to\label{step:fx}
			%$$\frac{\ln (k+1) \cdot }{c(p)}$$
			$$x_p(r(p,t)) \leftarrow \frac{1}{\mu} \cdot \left(     \exp \left(  \frac{\ln (\mu+1)}{c(p)} \cdot Y_p(r(p,t))\right)-1\right).$$

			\If{$x_p(r(p,t))-\frac{z_p(r(p,t))}{2}\geq \frac{1}{4 \cdot N \cdot \mu}$}{
				
				$z_p(r(p,t)) \leftarrow 2 \cdot x_p(r(p,t))$.\label{step:z}
				
			}

				\If{$x_p(r(p,t)) \geq \frac{1}{2}$}{
				
				$z_p(r(p,t)),x_p(r(p,t))  \leftarrow 1$.\label{step:1/2}
				
			}
			
			}

		}
				
			}

		}
		Return the (primal) solution $x$ and the (dual) solution $y$. 
\end{algorithm}

Once the primal constraint corresponding to $t,Q$ is satisfied, or $Q$ is no longer minimal for $x$ (a page $p \in \cP \setminus Q$ reaches $1$), there are two cases. If $x$ is feasible, the algorithm proceeds to the next time step. Otherwise, the algorithm repeats the above process with a new minimal violating set $Q'$. An additional property that will be useful in the analysis is that all non-zero entries of the obtained solution will be larger than $\Omega\left( \frac{1}{N \cdot \mu}\right)$ and there will be no fractional entries larger than $\frac{1}{2}$. Thus, we update through the algorithm another primal solution $z$; we update an entry $z_p(j)$ to be the value of $2 \cdot x_p(j)$ whenever the difference $x_p(j)-z_p(j)$ becomes larger than  $\Omega\left( \frac{1}{N \cdot \mu}\right)$. Moreover, we increase $z_p(j)$ %whenever they exceed the fraction of $\frac{1}{2}$. 
immediately to $1$ once $x_p(j)$ reaches $\frac{1}{2}$. The pseudocode of the algorithm is given in \Cref{alg:fractional}. 
%We first prove feasibility. 

\subsection{Analysis of \Cref{alg:fractional}}

We now analyze the competitive ratio of the fractional algorithm. We start by claiming the feasibility of the solutions obtained throughout the execution of the algorithm. 

	\begin{lemma}
	\label{lem:primalX}
	The primal solution $x$ defined by 	\Cref{alg:fractional} is feasible for the \textnormal{LP} \eqref{eq:LP2}. 
\end{lemma}

	\begin{lemma}
	\label{lem:zFeasible}
	The solution $z$ returned by \Cref{alg:fractional} is feasible for \eqref{eq:LP2}. 
\end{lemma}

	\begin{lemma}
	\label{lem:FFeasible}
	\Cref{alg:fractional} returns a feasible dual solution to the \textnormal{LP} \eqref{eq:LP2}. 
\end{lemma}

\comment{
\begin{proof}
To prove that the dual $y$ is also feasible, consider some $p \in \cP$ and $j \in [n_p]$. By the update rate of $x_p(j)$ in \Cref{step:fx} it follows that 
		\begin{equation}
		\label{eq:FromSAP}
		\frac{1}{\mu} \cdot \left(     \exp \left(  \frac{\ln (\mu+1)}{c(p)} \cdot Y_p(j)\right)-1\right) = x_p(j) \leq 1.
	\end{equation} 
	The inequality holds since $Y_p(j)$ does not increase in time $t$ if $p$ is already fully evicted at this time by \Cref{step:InWhile}; that is, once $x_p(r(p,t)) = 1$ it holds that $Q$ is no longer minimal and the algorithm chooses a new minimal set $Q'$ that includes $p$. Thus, by simplifying the expression in \eqref{eq:FromSAP} it holds that 
	$$\exp \left( \frac{\ln (\mu+1)}{c(p)} \cdot Y_p(j) \right) \leq \mu+1.$$
	by taking logarithms from both sides $Y_p(j) \leq c(p)$. Hence, $y$ is a feasible dual solution. 
\end{proof}
}
It remains to prove that the algorithm is $O(\log (\mu))$-competitive. To do so, we bound the increase in the primal $x$ by a $O(\log \mu)$ factor of the dual increase, at any time.

	\begin{lemma}
	\label{thm:logL1}
%	\Cref{alg:fractional} is $O(\log (n))$-\textnormal{competitive}. 
The cost of $x$ is bounded by $O(\log (\mu))$ times the value of the dual $y$. 
\end{lemma}

\begin{proof}
	Consider an infinitesimal increase in the value of the dual solution $y$. Specifically, assume that the algorithm chooses a minimal set $Q$ in \Cref{step:fQ} for time step $t$ and that the dual variable $y_t(\cP \setminus Q)$ increases infinitesimally by $d y_t(\cP \setminus Q)$. Let $dx$ and $dy$ denote the infinitesimal change in the objective value of $x$ and $y$, respectively. We bound the increase $dx$ in $x$ as a result of the increase $dy$.   
\begin{equation}
	\label{eq:dx}
	\begin{aligned}
		dx ={} & \sum_{p \in \cP-p_t} d x_p(r(p,t)) \cdot c(p) \\
		={} & \sum_{p \in (\cP-p_t) \setminus Q}  \frac{d x_p(r(p,t)) \cdot c(p) \cdot d y_t(\cP \setminus Q)}{d y_t(\cP \setminus Q)} \\
		={} & \sum_{p \in (\cP-p_t) \setminus Q} \ln \left(\mu+1 \right)\ \cdot \left( x_p(r(p,t))+\frac{1}{\mu} \right) \cdot d y_t(\cP \setminus Q)
	\end{aligned}
\end{equation} The first equality holds since the increase in $y_t(\cP \setminus Q)$ induces an increase only on the primal variables corresponding to pages in $(\cP-p_t) \setminus Q$ by \eqref{eq:LP2}. The last equality follows from the growth rate of a variable $x_p(r(p,t))$, for some $p \in \left(\cP-p_t\right) \setminus Q$, as a result of the growth in $y_t(\cP \setminus Q)$. %(for more details see~\cite{bansal2012primal,bansal2012randomized}). 
We separately analyze two of the expressions in \eqref{eq:dx}. First, since $y$ changes as a result of the increase in the variable $y_t(\cP \setminus Q)$, by \Cref{step:InWhile} it implies that 
\begin{equation}
	\label{eq:InwhileF}
	\sum_{p \in \left(\cP-p_t\right) \setminus Q} x_p(r(p,t)) \cdot g_Q(p) \leq 	\sum_{p \in \cP-p_t} x_p(r(p,t)) \cdot g_Q(p) < N-g(Q)
	%\sum_{p \in Q-p_t} x_p(r(p,t)) < 1.
\end{equation} %Second, 

For the second expression, let $S' \in \left(\cP-p_t\right) \setminus Q$ such that (i) $g_Q(S') \geq N-g(Q)$ and (ii) $S'$ is of minimum cardinality of all such sets. Clearly, there is such $S'$ as $S'' = \left(\cP-p_t\right) \setminus Q$ satisfies the first condition ($p_t$ is feasible alone in cache). Since $S'$ satisfies the cover constraints it holds that $f(\cP \setminus (S' \cup Q) \leq k$; thus, by the definition of $\mu$ it holds that $|\cP \setminus (S' \cup Q| \leq \mu$. Therefore,

\begin{equation}
\begin{aligned}
    	\label{eq:WidthIn}
		\sum_{p \in \left(\cP-p_t \right) \setminus Q} \frac{1}{\mu} ={} & \frac{ \left| \left(\cP-p_t \right) \setminus Q \right|-1}{\mu} 
  \\={} &
  \frac{|S'|+\left|\cP \setminus (S' \cup Q)\right|-1}{\mu} 
  \\\leq{} &
  \frac{N-g(Q)+\mu-1}{\mu} 
  \\\leq{} &
  N-g(Q)+1\leq 2 \cdot \left(N-g(Q)\right). 
  \end{aligned}
\end{equation} 
The first inequality holds since $|S'| \leq N-g(Q)$ because $S'$ is the minimum cardinality set that covers the demand of $N-g(Q)$; thus, the marginal contribution of any page in $S'$ to the cover is at least $1$ implying the inequality. The first inequality also uses $|\cP \setminus (S' \cup Q| \leq \mu$ as explained above.
For the last inequality, note that $N-g(Q) \geq 1$ since we assume that $y_t(\cP \setminus Q)$ increases at this time, implying that the corresponding constraint of $t$ and $S$ is not trivially satisfied. Therefore, by \eqref{eq:dx}, \eqref{eq:InwhileF}, and \eqref{eq:WidthIn},
\begin{equation}
	\label{eq:PDE}
	\begin{aligned}
	dx \leq{} &  \ln \left(\mu+1 \right) \cdot d y_t(\cP \setminus Q) \cdot \left(  	\sum_{p \in \left(\cP-p_t\right) \setminus Q} x_p(r(p,t)) + 	\sum_{p \in \left(\cP-p_t\right) \setminus Q} \frac{1}{\mu} \right) \\
	={} &  \ln \left(\mu+1 \right) \cdot d y_t(\cP \setminus Q) \cdot 3 \cdot \left( N-g(Q)\right)\\
	={} & O(\log (\mu)) \cdot dy.
	\end{aligned}
\end{equation} Thus, by \eqref{eq:PDE}, every increase in $y$ incurs an increase of at most a factor $O(\log (\mu))$ in $x$. Finally, note that if $x_p(j) \geq \frac{1}{2}$ then we immediately increase $x_p(j)$ to $1$; this increase the total cost of $x$ by a factor of $2$ w.r.t. the value of $y$. 
%Since $x$ and $y$ are feasible primal and dual solutions by \Cref{lem:FFeasible}, the proof follows from \Cref{lem:relaxation2}. 
\end{proof}

To conclude, by \Cref{step:z} and \Cref{step:1/2} we can trivially bound the cost of $z$ by a constant factor of the cost of $x$.

\begin{obs}
	\label{obs:z}
	The cost of $z$ is bounded by $4$ times the cost of $x$. 
\end{obs}

Finally, using the above we summarize the properties of $z$.

	\begin{lemma}
	\label{lem:ALGf}
	\Cref{alg:fractional} returns a feasible primal solution $z$ to \eqref{eq:LP2} such that the following holds. 
	\begin{enumerate}
		\item For all $p \in \cP$ and $j \in [n_p]$ it holds that either $z_p(j) \in \{0,1\}$ or that $z_p(j) \in \left[\frac{1}{4 \cdot N \cdot \mu},\frac{1}{2}\right]$.
		\item The cost of $z$ is bounded by $O(\log (\mu))$ times the cost of $\OPT$. 
	\end{enumerate}
\end{lemma}

\comment{
\begin{proof}
	The feasibility of $z$ follows from \Cref{lem:zFeasible}. Moreover, the first property follows from \Cref{step:1/2,step:z} of the algorithm. For the second property, since $x$ and $y$ are feasible primal and dual solutions by \Cref{lem:primalX} and \Cref{lem:FFeasible}, it follows that the cost of $x$ is bounded by $O(\log (\mu))$ times the cost of $\OPT$. Thus, the proof follows from Observation~\ref{obs:z}.  
\end{proof}
}

\subsection{Randomized Rounding}

In this section, we construct a randomized algorithm for supermodular paging based on an online rounding scheme of the solution $z$ to the LP \eqref{eq:LP2} obtained by \Cref{alg:fractional}. Let $$C = \max_{p,q \in \cP} \frac{c(p)}{c(q)}$$ be the maximum ratio of costs taken over all pairs pages; we assume without the loss of generality that all costs are strictly positive. As a scaling factor for our algorithm, let $\alpha = \log \left(4 \cdot C \cdot N^2 \cdot \mu^2\right)$ and let $z'$ be the solution obtained by augmenting $z$ by a factor of $\alpha$. That is, for all $p \in \cP$ and $j \in [n_p]$ define $z'_p(j) = \min \left( 1, \alpha \cdot z_p \right)$.  

Our algorithm computes $z'$ in an online fashion. At every time $t$, after updating $z'$, the algorithm evicts each page $p \in \cP$ from the cache with probability chosen carefully so that the total probability that $p$ is missing from cache after this moment is exactly $z'_p(r(p,t))$. %$\Delta_p$, which is the difference in the value of $z'_p(r(p,t))$ not yet exploited for the randomized rounding. 
If the cache is not feasible after this randomized rounding procedure, we evict pages that increase the total cover until we reach feasibility. The pseudocode of the algorithm is given in \Cref{alg:randomizedRounding}.

\begin{algorithm}[h]
	\caption{$\textsf{Randomized Rounding}$}
	\label{alg:randomizedRounding}

	\For{$t \in T$}{
		{\bf If} $p_t$ is missing from cache: fetch $p_t$. 
		
		Initialize $\Delta_{p_t} \leftarrow 0$. 
		
Update the solution $z$ according to \Cref{alg:fractional}.

Define $z'_p(r(p,t)) \leftarrow \min \left( 1, \alpha \cdot z_p \right)$ for all $p \in \cP$. 
		
			\For{$p \in \cP$}{
			Evict $p$ from cache with probability $\frac{z'_p(r(p,t))-\Delta_p}{1-\Delta_p}$.\label{step:Revict}
			
			Update $\Delta_p = z'_p(r(p,t))$.  
		}
		
		Let $\cF$ be the pages outside of the cache (part of the cover).\label{step:F}
		
		\While{$g(\cF)<N$\label{step:whileFe}}{
		Evict a page $p \in \cP \setminus \cF$ such that $g_{\cF}(p)>0$.\label{step:reach}
		}
		
	}
	
	Return the integral solution. 
\end{algorithm}

Observe that we evict a page $p$ at time $t$ with probability $\frac{z'_p(r(p,t))-\Delta_p}{\Delta_p}$, conditioned on the event that $p$ is still in the cache. Thus, the probability that $p$ belongs to the cache at the beginning of time $t$ is $\Delta_p$, and is $z'_p(r(p,t))$ after \Cref{step:Revict}. Thus, we have the following observation.

\begin{obs}
	\label{obs:z'}
	For all $p \in \cP$ and $t \in T$ the probability that $p$ is missing from the cache at time $t$ is at least $z'_p(r(p,t))$. 
\end{obs}

To analyze the performance of the algorithm, we use the following lemma. The proof follows from Lemma 2.5 in \cite{gupta2020online} combined with Observation~\ref{obs:z'}. 

	\begin{lemma}
	\label{lem:Gupta}
	Let $\cF$ be the set of pages outside cache at \Cref{step:F} at time $t$. Then, 
	$$\E \left[ g(\cF)\right] \geq N-e^{-\alpha} \cdot N \geq N-\frac{1}{2 \cdot \mu^2 \cdot C \cdot N}.$$
\end{lemma}

Using \Cref{lem:Gupta}, we bound the expected cost of the algorithm.%If after the randomized algorithm

	\begin{lemma}
	\label{lem:ALG3}
	\Cref{alg:randomizedRounding} returns a feasible integral solution for \textnormal{supermodular paging} with expected cost $O\left(\log (\mu) \cdot \alpha \right) \cdot \OPT = O\left(\log^2 (\mu) \cdot\log \left(C \cdot N\right)\right) \cdot \OPT$.  
\end{lemma}

The proof of \Cref{thm:randomized} follows from \Cref{lem:ALG3}. Moreover, the proof of \Cref{thm:approximation} follows using the ``round-or-separate''
approach of \cite{gupta2020online}.

\newpage

\bibliographystyle{splncs04}
%\newpage
\bibliography{bibfile}

\appendix

\section{Hardness Results}
\label{sec:hardness}

In this section we give hardness results for non-linear paging. In particular, we give a lower bound for supermodular paging based on a reduction from online set cover. This yields the proofs of \Cref{lem:Set Cover,thm:LB,thm:hardness}. We remark that the proof can be strengthened by providing a reduction from online submodular cover, however, we omit the details as we are not aware of stronger lower bounds for online submodular cover (compared to online set cover). In addition, in \Cref{sec:hardRestricted} we give a lower bound for solving non-linear paging in restricted cases. %online set cover and online submodular cover admit similar lower bounds \cite{alon2003online,korman2004use}.

 Recall that in online set cover we are given  a ground set $X = \{1,2,\ldots, n\} = [n]$ and a family $\cS = \{S_1,\ldots, S_m\}$ of subsets of $X$ (note that in this section only $n$ does not describes the number of pages). Requests for certain elements from $X$ arrive online; let $i_t$ be the requested element in time $t$ for a set of time steps $T$. If the requested element is not covered by a previously chosen set, we choose a set $S \in \cS$ containing the element to cover it, paying a cost $c(S)$. The goal is to minimize the cost of selected sets. Throughout this section, we fix the arbitrary online set cover instance $I$ described above. 

We construct a reduced supermodular paging instance $R_I$. Define the set of pages as $\cP = X \cup \cS$; that is, we define a page for every element and every set of the online set cover instance $I$. Define a cost $c(i) = \infty$ for every element $i \in [n]$ and for every set $S \in \cS$ we keep the original cost $c(S)$ from the set cover instance $I$ also as the cost for the supermodular paging instance $R_I$. Finally, define the cover function $g: 2^{\cP} \rightarrow \mathbb{N}$ such that for all $F \subseteq \cP$ 
\begin{equation}
	\label{eq:g}
	g(F) = \left| \{i \in X~|~ \exists S \in \cS \cap F \text{ s.t. } i \in S\} \cup \left(X \cap F\right) \right|. 
\end{equation} In simple words, $g(F)$ describes the number of elements $i$ in $X$ {\em covered} either by a set $S$ (i.e., $i \in S$) that belongs to $F$ or by the actual corresponding element that belongs to $F$ (i.e., $i \in F$). %An illustration of the construction is given in \Cref{fig:Y}. 

 Clearly, $g$ is a cover function implying that it is submodular. Therefore, define the corresponding feasibility of $R_I$ as function $f(A) = n-g(\cP \setminus A)~\forall A \subseteq \cP$. It follows that $f$ is supermodular. Define the cache threshold as $k = 0$. Thus, a feasible set of pages $F$ can be in the cache if and only if we cover all elements (either by a set or by the element itself).

\begin{obs}
	\label{obs:F}
	For every $F \subseteq \cP$, it holds that $f(F) \leq k$ if and only if for every $x \in X$ there is $p \in \cP \setminus F$ such that (i) $p \in \cS$ and $x \in p$ or (ii) $p \in X$ and $x = p$. 
\end{obs}

 %$f(F) \leq k$; this happens if and only if $g(F) \geq n$; that is, if all elements in $X$ are covered. 

We are left with setting the sequence of page requests. We define two consecutive subsequences of requests. The first subsequence $Q_1 = (p_1,\ldots, p_m)$ requests at time $1 \leq j \leq m$ the page $p_j = S_j$. Observe that for every time $1 \leq j \leq m$ it holds that the set of pages $F_j$ in cache is feasible because $F_j \cap X = \emptyset$; thus, every element $x \in X$ covers itself w.r.t. $g$, i.e., $g(F_j) \geq g(X) = n$. 
%recall that $x \in X$ is corresponding page to the element $x$ from the set cover instance $I$. 
For the second subsequence, assume without the loss of generality that $T \cap [m] = \emptyset$, i.e., the set of time points requested by the set cover instance is disjoint to the set of time points used for our first subsequence $Q_1$. Now, define the second subsequence $Q_2 = \left(i_t  ~|~ t \in T\right)$ as the set of element requests by the set cover instance $I$. %for every $t \in T$ define $p_t = i_t$
An illustration of the construction is given in \Cref{fig:Y}.

We use the following result. In the proof, we construct a solution for the set cover instance $I$ based on a solution for $R_I$.

%\begin{lemma}
%	\label{lem:Set Cover}
%	For any $\rho \geq 1$, if there is a $\rho$-competitive algorithm for \textnormal{supermodular paging} then there is a $\rho$-competitive algorithm for \textnormal{online set cover} of the same running time up to a polynomial factor.  
%\end{lemma}

\subsubsection*{Proof of \Cref{lem:Set Cover}}

For some $\rho \geq 1$,	let $\cA$ be a $\rho$-competitive algorithm for \textnormal{supermodular paging}. Given instance $I$ of online set cover, construct the reduced instance $R_I$ and apply algorithm $\cA$ on $R_I$ in an online fashion. At time $t \in T$, define $H_t \subseteq \cS$ as the collection of sets that are not in the cache at time $t$. As $\cP \setminus H_t$ must be feasible in the cache by the feasibility of $\cA$, it holds that $f(\cP \setminus H_t) \leq k$; thus, by \Cref{obs:F} every element in $X$ is either (i) covered by some set $S \in \cS$ or (ii) not yet requested by the set cover instance. Hence, $H = (H_t~|~t \in T)$ forms a feasible solution for $I$.  
	
	As $\cA$ is $\rho$-competitive for a constant $\rho$, it holds that $\cA$ does not evict a page $x \in X$ as it would pay an infinite cost. In addition, as the costs of sets are the same in $I$ and $R_I$, we conclude that the cost paid by $\cA$ on $R_I$ is the same as the cost paid by the solution $H$ for $I$. Moreover, the optimum of $R_I$ is at least as large as the optimum of $I$: at every moment $t \in T$, since $\cA$ does not evict pages in $X$, it must use only sets to cover the requested elements from $Q_2$, incuring the costs of the sets.  Thus, the above gives a $\rho$-competitive algorithm for online set cover with the same running time up to a polynomial factor for constructing the reduction. \qed

 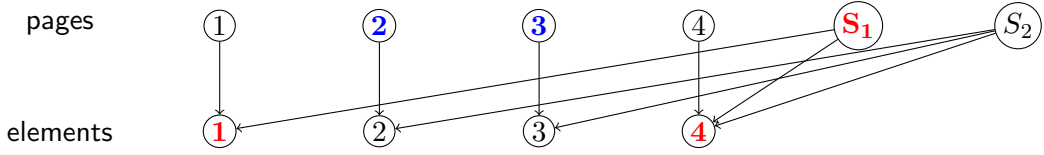
\begin{figure}
	%	\hspace{4cm}{
		\centering
		\begin{tikzpicture}[scale=1.4, every node/.style={draw, circle, inner sep=1pt}]
			% first bipartite graph
			\node (p2) at (5.5,-0.5) {$\bf \textcolor{blue}{2}$};
			\node (p1) at (4,-0.5) {$\textcolor{black}{1}$};
			\node (p3) at (7,-0.5) {$\bf \textcolor{blue}{3}$};
			\node (a1) at (4,-1.5) {$\bf \textcolor{red}{1}$};
			
			\node (a2) at (5.5,-1.5) {$\textcolor{black}{2}$};
			\node (a3) at (7,-1.5) {$\textcolor{black}{3}$};
			\node (a4) at (8.5,-1.5) {$\bf \textcolor{red}{4}$};
			\node (p4) at (8.5,-0.5) {$4$};
			
				\node (S1) at (10,-0.5) {$\bf \textcolor{red}{S_1}$};
				
					\node (S2) at (11.5,-0.5) {$S_2$};
	\draw[->] (S1) -- (a1);
	
		\draw[->] (S2) -- (a3);
			\draw[->] (S2) -- (a4);
				\draw[->] (S2) -- (a2);
				
		\draw[->] (S1) -- (a4);			
			\draw[->] (p1) -- (a1);
				\draw[->] (p2) -- (a2);
					\draw[->] (p3) -- (a3);
						\draw[->] (p4) -- (a4);
			%		\draw (p3) -- (a2);
			%	\draw[line width=2pt, color=red] (p3) -- (a3);
			
%			\draw[->] (p3) -- (a3);
%			\draw[->] (p2) -- (a1);
%			\draw[->] (p1) -- (a2);
%			\draw[->] (p2) -- (a3);
%			
%			\draw[->] (p4) -- (a2);
%			\draw[->] (p4) -- (a3);
%			\draw[->] (p4) -- (a4);

			\node[draw=none] at (2.5, -0.5) {$\textsf{pages}$};
			
			\node[draw=none] at (2.5, -1.5) {$\textsf{elements}$};
			
		\end{tikzpicture}
		%\vspace{-1.5cm} 
		\caption{\label{fig:Y} An illustration of the construction for $X = \{1,\ldots,4\}$ and $\cS = \{S_1,S_2\}$. The considered moment in time takes place after a request for element-pages $1,4$; to serve these requests, the set-page $S_1$ (in red) is evicted from cache (taken to the cover). Observe that $\{2,3,S_1\}$, the pages outside of the cache, form a feasible cover of the elements $\{1,\ldots,4\}$.}
	\end{figure}
 
%\begin{lemma}
%	\label{lem:Set Cover}
%	For any $\rho \geq 1$, if there is a $\rho$-competitive algorithm for \textnormal{supermodular paging} then there is a $\rho$-competitive algorithm for \textnormal{online set cover} of the same running time up to a polynomial factor.  
%\end{lemma}
Observe that in our reduction $|\cP| = m+n$ and that $f(\cP) = n$. Thus, by \Cref{lem:Set Cover} and the results of \cite{alon2003online,korman2004use}, we have the statement of \Cref{thm:hardness}. %(for $f(\cP) = \Omega(n)$).

\subsubsection*{Proof of \Cref{thm:LB}}

By the results of \cite{korman2004use}, there are constants $a,b>0$ such that there are online set cover instances with $n$ elements, $m = n^a$ sets such that the minimum number of sets required to form a feasible cover at the end of the algorithm is $K = n^b$ and no randomized online algorithm is $\Omega(\log^2\left(n\right))$-competitive on these instances unless $\textnormal{NP}\subseteq \textnormal{BPP}$. Moreover, note that the cardinality of the maximum minimally infeasible set (see \Cref{def:Mininfeasible}) is at least $K = n^b$ and at most $n$; thus, $\ell$ is polynomial in $n$ i.e., there is a constant $c>0$ such that $\ell = n^c$. Therefore, the lower bound on the running time obtained by \Cref{lem:Set Cover} combined with \cite{korman2004use} is $\Omega(\log^2\left(\ell\right)) = \Omega(\log^2\left(n\right))$. \qed

\subsection{Hardness of Restricted Non-linear Paging instances}
\label{sec:hardRestricted}

We now show that even for very restricted cases of the problem, the competitive ratio is effectively unbounded (if the competitive ratio is not stated as a function of $\ell$ or $n$). 

\begin{theorem}
\label{thm:HardnessRestricted}
For every $\rho \geq 1$, there is no randomized or deterministic $\rho$-competitive algorithm for \textnormal{non-linear paging} even if $k = 0$ and the range of $f$ is $\{0,1\}$.
\end{theorem}

\begin{proof}
We give a reduction from the classic paging problem. Let $I$ be a paging instance with a set $\cP$ of $n+1$ pages, cache capacity $n$, and requests $p_t$ for every point of time $t \in T$. We define a non-linear paging instance $R$ as follows. The set of pages and requests are the same in $I$ and $R$. Define a feasibility function $f:\cP \rightarrow \{0,1\}$ such that for every $S \subseteq \cP$ it holds that $f(S) = 0$ if $|S| \leq n$ and $f(S) = 1$ otherwise. Finally, define $k = 0$ as the cache size constraint of $R$. Clearly, a subset of pages $S \in \cP$ is feasible for $I$ if and only if it is feasible for $R$. Thus, by the well-known hardness results for deterministic paging, we cannot obtain better than $n$-competitive \cite{sleator1985amortized} (recall that the instance has $n+1$ pages and the cache capacity is $n$); if randomization is allowed, we cannot obtain better than $O(\log n)$-competitive \cite{fiat1991competitive}. As $n$ can be arbitrarily large, we prove the theorem.  
    %Fix some integer $\rho \geq 1$
\end{proof}

\section{Fractional Algorithm for Strengthened LP}%\textsf{Stronger-LP}}
\label{sec:stronger}

The integrality gap example shows that LP \eqref{eq:LP} is not sufficient for obtaining a randomized $\textnormal{polylog}(\ell)$-competitive algorithm for general non-linear paging. Instead, we describe a stronger version of our LP \eqref{eq:LP}, %
in which we require to remove from each infeasible set $S$ a set of pages $S'$, so that the complement of $S'$ in $S$, $S \setminus S'$, will be feasible in the cache (i.e., $ f\left(S \setminus S'\right) \leq k$).  
%As we do not know a priori which such set $S'$ is evicted by the integral optimum, we only demand that evicting a minimum number of pages from $S$ whose complement size fits in cache. 
Formally, for a set $S \subseteq \cP$, let $\cC(S) = \left\{S' \subseteq S \text{ s.t. } f\left(S \setminus S'\right) \leq k\right\}$ be the collection of subsets of $S$ such that $ f\left(S \setminus S'\right) \leq k$ and define
\begin{equation}
	\label{eq:q(S)}
	q(S) = \min_{S' \in \cC(S)} \left|S'\right|
\end{equation}
as the number of pages needed to be evicted from $S$ at any point in time. 
The LP is given below using similar notation to \eqref{eq:LP}.

\begin{equation}
	\label{eq:LPS}
	\begin{aligned}
		\textsf{Stronger-LP}:~~~~~~~~	& \min \sum_{p \in \cP~} \sum_{j \in [n_p]}  x_p(j) \cdot c(p)\\
		& ~~~~~~~~\text{s.t. }\\
		& \sum_{p \in S-p_t} x_p(r(p,t)) \geq q(S), 	~~~~~~~~~~~~~~\forall t \in T~\forall S  \in \cS(t) \\ 
		& x_p(j) \geq 0 ~~~~~~~~~~~~~~~~~~~~~~~~~~~~~~~~~~~	\forall p \in \cP~\forall j \in [n_p]
	\end{aligned}
\end{equation} For example, in classic paging, for a subset of pages $S \subseteq \cP$ it holds that $q(S) = |S|-k$, and hence the above constraints require removing at least $|S|-k$ pages from cache within $S$, coinciding with the LP used for solving (weighted) paging \cite{bansal2012primal}. % at least for classic paging, the integrality gap of the LP does not exceed $O(\log k)$. 
Clearly, LP \eqref{eq:LPS} is a relaxation of non-linear paging, since an integral solution, for any infeasible set $S$, must evict at least $q(S)$ pages from $S$.

In this section, we give a fractional $O\left(\log \mu\right)$-competitive algorithm for solving \textsf{Stronger-LP} given in \eqref{eq:LPS}, which yields the proof of \Cref{thm:MU}. The dual LP of \eqref{eq:LPS} is as follows. 

\begin{equation}
	\label{eq:dualS}
	\begin{aligned}
		%\textsf{Dual-LP}:
		~~~~~~~~	& \max~~~ \sum_{t \in T~} \sum_{S  \in \cS(t)} y_t(S) \cdot q(S)\\
		& ~~~~~~~~\text{s.t. }\\
		& \sum_{t \in I(p,j)~} \sum_{S \in \cS(t)\big|p \in S-p_t} y_t(S) \leq c(p), 	~~~~~~~~~~~~~~\forall p \in \cP~\forall j \in [n_p]\\ 
		%	& y_t(S) \geq 0 ~~~~~~~~~~~~~~~~~~~~~~~~~~~~~	\forall t \in T(I) ~\forall S \in \bbc. 
	\end{aligned}
\end{equation}

Let $\textsf{Primal-LP}$ and $\textsf{Dual-LP}$  denote the optimum values of optimal solutions for the primal and the dual programs, respectively.  %, and by $\OPT$ the offline (integral) optimum.  
Our fractional algorithm initializes the (infeasible) primal solution and the (feasible) dual solutions as vectors of zeros $\bar{0}$. Upon the arrival of the requested page $p_t$ at time $t$ we do the following process until $x$ satisfies all constraints of the LP \eqref{eq:LPS} up to time $t$.

\begin{algorithm}[h]
	\caption{$\textsf{Fractional}$}
	\label{alg:fractionalS}
	
	%	Initialize the cache $\cG_0 \leftarrow \emptyset$.
	Initialize (infeasible) primal solutions $x,z \leftarrow \bar{0}$ 
	
	Initialize (feasible) dual solution $y \leftarrow \bar{0}$.

	\For{$t \in T$}{
		
		\While{$x$ is not feasible for $t$\label{step:Whilefeasible}}{
			
			%	Let $\textsf{Frac}_t = \left\{p \in \cP~|~r(p,t) \geq 1 \textnormal{ and } x_p(r(p,t)) < 1\right\}$ be pages not fully evicted.   
			
			Find $S \subseteq \cS(t)$ %of maximum cardinality 
			such that $\sum_{p \in S-p_t} x_p(r(p,t)) < q(S)$ and $x_p(r(p,t))<1~\forall p \in S-p_t $.\label{step:fQS}

			\While{$\sum_{p \in S-p_t} x_p(r(p,t)) < q(S)$ \textnormal{and} $x_p(r(p,t))<1~\forall p \in S-p_t $\label{step:InWhileS}}{
				
				Increase $y_t(S)$ continuously.\label{step:fyS}

				\ForAll{$p \in S-p_t$}{
					
					increase $x_p(r(p,t))$ according to\label{step:fxS}
					%$$\frac{\ln (k+1) \cdot }{c(p)}$$
					$$x_p(r(p,t)) \leftarrow \frac{1}{\mu} \cdot \left(     \exp \left(  \frac{\ln (n+1)}{c(p)} \cdot Y_p(r(p,t))\right)-1\right).$$

%					\If{$x_p(r(p,t))-\frac{z_p(r(p,t))}{2}\geq \frac{1}{4 \cdot N \cdot n}$}{
%						
%						$z_p(r(p,t)) \leftarrow 2 \cdot x_p(r(p,t))$.\label{step:zS}
%						
%					}
%					
%					
%					\If{$x_p(r(p,t)) \geq \frac{1}{2}$}{
%						
%						$z_p(r(p,t)),x_p(r(p,t))  \leftarrow 1$.\label{step:1/2S}
%						
%					}
					
				}

			}
			
		}

	}
	Return the (primal) solution $x$ and the (dual) solution $y$. 
\end{algorithm}

In time $t$, the algorithm considers a {\em violating set} of pages $S \subseteq \cP$. This set violates the primal constraint in \eqref{eq:LP2} corresponding to $t$ and $S$. %and is of maximum cardinality amongst all such sets.   %We remove all pages in $Q$ except $p_t$ from cache and 
We increase the variable $y_{t}(S)$ continuously and at the same time increase variables $x_p(r(p,t))$ for all pages in $S$ (except for $p_t$) that are not fully evicted yet. The increase rate is a function of $Y_p(r(p,t))$, where %recall that $Y_p(r(p,t))$
\begin{equation}
	\label{eq:Y2}
	Y_p(j) =\sum_{t \in I(p,j)~} \sum_{S \subseteq \cP \big|p \in S-p_t} y_t(S)
\end{equation}  is the left hand side of the corresponding dual constraint of $p$ and $j = r(p,t)$ in \eqref{eq:LP2} (analogously to \eqref{eq:Y}). This growth function has an exponential dependence on $Y_p(r(p,t))$, scaled by the cost of page $c(p)$, the number of pages, and the total cover demand. The growth of variable $x_p(j)$ stops once it reaches $1$. %$\frac{1}{2}$.   

%To compute $x_p(j)$, we continuously increase the (dual) variable $y_t(Q)$, corresponding to the non-valid set of pages in the cache at time $t$. At the same time, we increase the probability $x_p(t)$ (and equivalently, $X_p(t)$) according to an exponential function that depends on $y_p(j)$, where recall that $y_p(j)$ is the sum over all dual variables corresponding to $p$ in the interval $R(p,j)$. Once a valid distribution is constructed, we sample a page according to $x(t)$ and achieving a valid cache state once again. We also remove from the cache pages $p'$ for which $X_{p'}(t) \geq 1$. The pseudocode of the algorithm is given in \Cref{alg:fractional}. 

Once the primal constraint corresponding to $t,S$ is satisfied, there are two cases. If $x$ becomes feasible, the algorithm proceeds to the next time step. Otherwise, the algorithm repeats the above process with a new set $S'$ for which $x$ does not satisfy its constraint. The pseudocode of the algorithm is given in \Cref{alg:fractionalS}.

\subsubsection*{Analysis of \Cref{alg:fractionalS}}
In the following we analyze the competitive performance of the algorithm.  We first prove the feasibility of the primal and dual solutions $x,y$. %Note that the algorithm returns a feasible primal solution $x$ by \Cref{step:InWhileS}. 

	\begin{lemma}
	\label{lem:PFeasibleS}
	\Cref{alg:fractionalS} returns a feasible primal solution $x$ to \eqref{eq:LPS}. 
\end{lemma}

\begin{proof}
	Let $t \in T$ and $S \in \cS(t)$. We show that $x$ satisfies the corresponding primal constraint of $t,S$. Let $U = \left\{p \in S~|~x_p(r(p,t)) = 1\right\}$ be the set of fully-evicted pages in $S$ (observe that $p_t \notin U$ since $x_{p_t}(r(p,t)) = 0$). We use the following auxiliary claims. 
	
		\begin{claim}
		\label{claim:r<44}
		%$\cS\left(S \setminus U\right) \subseteq \left\{S' \in \cS(S)~|~S' \setminus U\right\}$
		For every $S' \in \cC(S \setminus U)$ it holds that $S' \cup U \in \cC(S)$.
	\end{claim}
	\begin{claimproof}
		It holds that 
		$$f \left(S \setminus \left(S' \cup U\right)\right) = f \left((S \setminus U) \setminus S'\right) \leq k.$$
		The inequality holds since $S' \in \cC(S \setminus U)$ (see the definition of the sets $\cC(S),\cC(S \setminus U)$ above \eqref{eq:q(S)}). By the definition of $\cC(S)$ it holds that $S' \cup U \in \cC(S)$.
	\end{claimproof}
	The following claim relies on Claim~\ref{claim:r<44}. 
		\begin{claim}
		\label{claim:r<45}
		%$\cS\left(S \setminus U\right) \subseteq \left\{S' \in \cS(S)~|~S' \setminus U\right\}$
	$q(S) \leq |U|+q(S \setminus U)$. 
	\end{claim}
	\begin{claimproof}
		It holds that 
		\begin{equation*}
			\begin{aligned}
			q(S) ={} & \min_{S' \in \cC(S)} |S'| 
			\\={} & \min_{S' \in \cC(S)} \left( |S' \cap U|+|S' \setminus U| \right) 
			\\\leq{} &  
			\min_{S' \in \cC(S) \text{ s.t. } U \subseteq S'} \left( |S' \cap U|+|S' \setminus U| \right) 
				\\={} &  
				|U|+\min_{S' \in \cC(S) \text{ s.t. } U \subseteq S'}|S' \setminus U| 
					\\\leq{} &  
				|U|+\min_{S' \in \cC(S \setminus U)} |S'| 
					\\={} &  
				|U|+q(S \setminus U). 
			\end{aligned}
		\end{equation*} The first inequality holds by restricting the minimum to a subset of the set $\cC(S)$; note that there is $S' \in \cC(S)$ such that $U \subseteq S'$ since $S \in \cC(S)$ and $U \subseteq S$. The second inequality follows from Claim~\ref{claim:r<44}. 
	\end{claimproof}
	
	Therefore, 
	$$\sum_{p \in S-p_t} x_p(r(p,t)) = \sum_{p \in U} x_p(r(p,t))+\sum_{p \in (S\setminus U)-p_t} x_p(r(p,t)) \geq |U|+q(S \setminus U) \geq q(S).$$
	The first inequality holds by the following. First, since every page in $U$ is fully evicted. Second, since $$\sum_{p \in (S\setminus U)-p_t} x_p(r(p,t)) \geq q(S \setminus U)$$ as non of the pages in $S \setminus U$ is fully-evicted; thus, the constraint of $t,S-U$ must be satisfied by the stopping condition of \Cref{step:InWhileS}.  The last inequality follows from Claim~\ref{claim:r<45}. 
\end{proof}

We now show the feasibility of the dual solution $y$ using the growth rate of the variables in the algorithm. 
	\begin{lemma}
	\label{lem:FFeasibleS}
	\Cref{alg:fractionalS} returns a feasible dual solution to \eqref{eq:LPS}. 
\end{lemma}

\begin{proof}
	%To prove that the dual $y$ is feasible, 
	Consider some $p \in \cP$ and $j \in [n_p]$. By the update rate of $x_p(j)$ in \Cref{step:fxS} it follows that 
	\begin{equation}
		\label{eq:FromSAPS}
		\frac{1}{\mu} \cdot \left(     \exp \left(  \frac{\ln (\mu+1)}{c(p)} \cdot Y_p(j)\right)-1\right) = x_p(j) \leq 1.
	\end{equation} 
	The inequality holds since $Y_p(j)$ does not increase in time $t$ if $p$ is already fully evicted at this time by \Cref{step:InWhileS}. Thus, by simplifying the expression in \eqref{eq:FromSAPS} it holds that 
	$$\exp \left( \frac{\ln (\mu+1)}{c(p)} \cdot Y_p(j) \right) \leq \mu+1.$$
	by taking $\ln$ over both sides it follows that $Y_p(j) \leq c(p)$. Hence, $y$ is a feasible dual solution. 
\end{proof}

We finally prove the competitive ratio of the algorithm.

%	\begin{lemma}
%	\label{thm:logL1S}
%	%	\Cref{alg:fractional} is $O(\log (n))$-\textnormal{competitive}. 
%	The cost of $x$ is bounded by $O(\log (\mu))$ times the value of the dual $y$.  
%\end{lemma}

\subsubsection*{Proof of \Cref{thm:MU}}

\begin{proof}
	Consider an infinitesimal increase in the value of the dual solution $y$. Specifically, assume that the algorithm chooses a set $S$ in \Cref{step:fQS} for time time step $t$, and that the dual variable $y_t(S)$ increases infinitesimally by $d y_t(S)$. Let $dx$ and $dy$ denote the infinitesimal change in the objective value of $x$ and $y$, respectively. We bound the increase $dx$ in $x$ as a function of the increase $dy$.   
	\begin{equation}
		\label{eq:dxS}
		\begin{aligned}
			dx ={} & \sum_{p \in S-p_t} d x_p(r(p,t)) \cdot c(p) \\
			={} & \sum_{p \in S-p_t}  \frac{d x_p(r(p,t)) \cdot c(p) \cdot d y_t(S)}{d y_t(S)} \\
			={} & \sum_{p \in S-p_t} \ln \left(\mu+1 \right)\ \cdot \left( x_p(r(p,t))+\frac{1}{\mu} \right) \cdot d y_t(S). 
		\end{aligned}
	\end{equation} The first equality holds since the increase in $y_t(S)$ induces an increase only for the primal variables corresponding to pages in $S-p_t$. The last equality follows from the growth rate of a variable $x_p(r(p,t))$, for some $p \in S-p_t$ as a result of the growth in $y_t(S)$. %(for more details see~\cite{bansal2012primal,bansal2012randomized}). 
	We separately analyze the two expressions in \eqref{eq:dxS}. First, since at the considered moment in time $y_t(S)$ increases, by \Cref{step:InWhileS} it holds that 
	\begin{equation}
		\label{eq:InwhileFS}
		\sum_{p \in S-p_t} x_p(r(p,t)) \leq q(S)%\cdot g_Q(p) \leq 	\sum_{p \in \cP-p_t} x_p(r(p,t)) \cdot g_Q(p) < N-g(Q)
		%\sum_{p \in Q-p_t} x_p(r(p,t)) < 1.
	\end{equation} 
	
	For the second expression, let $S' \in \cC(S)$ such that $|S'| = q(S)$; clearly, by \eqref{eq:q(S)} there is such $S'$. By the definition of $\cC(S)$ it holds that $f(S \setminus S') \leq k$; thus, $|S \setminus S'| \leq \mu$. Therefore,  
	\begin{equation}
		\label{eq:WidthInS}
		\sum_{p \in S-p_t} \frac{1}{\mu} = \frac{ \left| S \right|-1}{\mu} = \frac{ \left| S' \right|+\left|S \setminus S'\right|-1}{\mu} \leq \frac{ q(S)+\mu-1}{\mu} \leq q(S)+1 \leq 2 \cdot q(S).  
	\end{equation} For the last inequality, note that $q(S) \geq 1$ since we assume that $y_t(S)$ increases implying that the corresponding constraint of $t$ and $S$ is not satisfied. Therefore, by \eqref{eq:dxS}, \eqref{eq:InwhileFS}, and~\eqref{eq:WidthInS},
	\begin{equation}
		\label{eq:PDESS}
		\begin{aligned}
			dx \leq{} &  \ln \left(\mu+1 \right) \cdot d y_t(S) \cdot  \left(	\sum_{p \in S-p_t} x_p(r(p,t)) + 	\sum_{p \in S-p_t} \frac{1}{\mu} \right) \\
			={} & 3 \cdot \ln \left(\mu+1 \right) \cdot d y_t(S) \cdot q(S)\\
			={} & O(\log (\mu)) \cdot dy.
		\end{aligned}
	\end{equation} Thus, by \eqref{eq:PDESS}, every increase in $y$ increases $x$ by a factor of at most $O(\log (\mu))$. By \Cref{lem:FFeasibleS,lem:PFeasibleS} the algorithm also returns feasible primal and dual solutions, implying the statement of the theorem.  
	%Finally, note that if $x_p(j) \geq \frac{1}{2}$ then we immediately increase $x_p(j)$ to $1$; this increase the total cost of $x$ by a factor of $2$ w.r.t. the value of $y$. 
	%Since $x$ and $y$ are feasible primal and dual solutions by \Cref{lem:FFeasible}, the proof follows from \Cref{lem:relaxation2}. 
\end{proof}

\section{Omitted Proofs}
\label{sec:omit}

\subsubsection*{Proof of \Cref{claim:z}}
	The proof resembles the proof of Claim 3.10 in \cite{coester2022competitive}. Assume that there are $t \in T$ and $S \subseteq \cP$ such that $x'$ violates the constraint of the LP \eqref{eq:LP2} corresponding to $t$ and  $S$ and that there is $q \in \cP \setminus S$ such that $x'_q(r(q,t)) = 1$. We show that $x'$ also violates the constraint corresponding to $t$ and $S+q$.  %As $x'$ satisfies all minimal constraints, there is 
	\begin{equation*}
		\label{eq:PQPQPQ}
		\begin{aligned}
	 N-g(S) >{} & \sum_{p \in \cP-p_t} x'_p(r(p,t)) \cdot g_S(p)
	 = x'_q(r(q,t)) \cdot g_S(q)+\sum_{p \in \cP-p_t-q} x'_p(r(p,t)) \cdot g_S(p)\\
	 ={} & g_S(q)+\sum_{p \in \cP-p_t-q} x'_p(r(p,t)) \cdot g_S(p)
	  \geq g_S(q)+\sum_{p \in \cP-p_t-q} x'_p(r(p,t)) \cdot g_{S+q}(p).\\
		\end{aligned}
	\end{equation*} The first inequality holds since we assume that $x'$ violates the primal constraint corresponding to $t$ and $S$. The second equality follows because $x'_q(r(q,t)) = 1$. The second inequality follows from the submodularity of $g$. As $g(S+q) = g(S)+g_{S}(q)$, by the above inequality we conclude that $x'$ also violates the constraint corresponding to $t$ and $S+q$. Hence, if $x'$ satisfies all minimal constraints it also satisfies all non-minimal constraints. \qed

\subsubsection*{Proof of \Cref{lem:primalX}}
 By \Cref{claim:z} we only need to show that $x$ satisfies all minimal constraints. %Observe that \Cref{step:Whilefeasible} along with the submodularity of $g$ ensures that the algorithm returns a feasible (fractional) solution $x$ to~\eqref{eq:LP2}. Specifically, 
If $x$ is not feasible at time $t$, the algorithm is guaranteed to find a minimal set $Q$ in \Cref{step:fQ} such that the primal constraint for $t$ and $Q$ is not satisfied. Then, the algorithm increases the variables of pages in $(\cP-p_t) \setminus Q$ until they satisfy the constraint of $Q,t$ by  \Cref{step:Whilefeasible} or until $Q$ is no longer minimal; as we only need to prove that $x$ satisfies all minimal constraints by \Cref{claim:z}, either the constraint corresponding to $t,Q$ is satisfied, or this constraint is not minimal anymore. %(i.e., $Q$ is not minimal). 
Thus, $x$ satisfies all minimal constraints, yielding the proof. \qed

\subsubsection*{Proof of \Cref{lem:zFeasible}}
 	By \Cref{claim:z} we only need to show that $z$ satisfies all minimal constraints of \eqref{eq:LP2}. Let $t \in T$ and let $S \subseteq \cP$ be such that $t$ and $S$ form a minimal constraint for $z$ in  \eqref{eq:LP2}. If $g(S) = N$ then the constraint of $t$ and $S$ is trivially satisfied by $z$; thus, assume that $N-g(S) \geq 1$. As the constraint of $t$ and $S$ is minimal for $z$, for all $p \in \cP\setminus S$ it holds that $z_p(r(p,t)) < 1$; thus, by \Cref{step:1/2} it follows that $z_p(r(p,t)) \leq \frac{1}{2}$. Hence,  by \Cref{step:z} for all $p \in \cP\setminus S$ it holds that 
 $z_p(r(p,t)) \geq 2 \cdot x_p(r(p,t))-\frac{1}{4 \cdot N \cdot \mu}$.
 Therefore,
\begin{equation}
	\label{eq:z2}
	\begin{aligned}
		 \sum_{p \in \cP-p_t} g_S(p) \cdot z_p(r(p,t))
		\geq{} &   \sum_{p \in \cP-p_t} g_S(p) \cdot \left( 2 \cdot x_p(r(p,t))-\frac{1}{4 \cdot N \cdot \mu} \right)\\
		={} & 2 \cdot \sum_{p \in \cP-p_t} g_S(p) \cdot x_p(r(p,t))-\sum_{p \in \cP-p_t} g_S(p) \cdot \frac{1}{4 \cdot N \cdot \mu} \\
		\geq{} & 2 \cdot \sum_{p \in \cP-p_t} g_S(p) \cdot x_p(r(p,t))-N \cdot \mu \cdot \frac{1}{4 \cdot N \cdot \mu} \\
			\geq{} & 2 \cdot \left(N-g(S)\right)-\frac{1}{4} 
				\geq N-g(S). 
	\end{aligned}
\end{equation}
%The first inequality holds by \eqref{eq:A1}. 
The second inequality holds since the marginal contribution of any page $p$ to $S$ is bounded by the total cover demand $N$, i.e., $g_S(p) \leq N-g(S) \leq N$. The third inequality holds since $x$ is a feasible solution for the LP by \Cref{lem:primalX}.   The last inequality follows from the assumption $N-g(S) \geq 1$. By \eqref{eq:z2} the proof follows. \qed

\subsubsection*{Proof of \Cref{lem:FFeasible}}
To prove that the dual $y$ is also feasible, consider some $p \in \cP$ and $j \in [n_p]$. By the update rate of $x_p(j)$ in \Cref{step:fx} it follows that 
		\begin{equation}
		\label{eq:FromSAP}
		\frac{1}{\mu} \cdot \left(     \exp \left(  \frac{\ln (\mu+1)}{c(p)} \cdot Y_p(j)\right)-1\right) = x_p(j) \leq 1.
	\end{equation} 
	The inequality holds since $Y_p(j)$ does not increase in time $t$ if $p$ is already fully evicted at this time by \Cref{step:InWhile}; that is, once $x_p(r(p,t)) = 1$ it holds that $Q$ is no longer minimal and the algorithm chooses a new minimal set $Q'$ that includes $p$. Thus, by simplifying the expression in \eqref{eq:FromSAP} it holds that 
	$$\exp \left( \frac{\ln (\mu+1)}{c(p)} \cdot Y_p(j) \right) \leq \mu+1.$$
	by taking logarithms from both sides $Y_p(j) \leq c(p)$. Hence, $y$ is a feasible dual solution. \qed

\subsubsection*{Proof of \Cref{lem:ALGf}}
The feasibility of $z$ follows from \Cref{lem:zFeasible}. Moreover, the first property follows from \Cref{step:1/2,step:z} of the algorithm. For the second property, since $x$ and $y$ are feasible primal and dual solutions by \Cref{lem:primalX} and \Cref{lem:FFeasible}, it follows that the cost of $x$ is bounded by $O(\log (\mu))$ times the cost of $\OPT$. Thus, the proof follows from Observation~\ref{obs:z}. \qed

\subsubsection*{Proof of \Cref{lem:ALG3}}
The feasibility of the algorithm follows from \Cref{step:whileFe}. To bound the expected cost of the algorithm, consider some time step $t \in T$. We bound the expected cost paid at time $t$ by the algorithm on evictions, w.r.t. the cost paid by $z'$ on fetching. Let $z'_p(t-1)$ be the value of the variable $z'_{p_t}(r(p_t,t))$ at time $t-1$, i.e., indicating the portion of $p$ missing from cache at time $t-1$. If $z'_p(t-1) = 0$, i.e., $p_t$ is in the cache when it is requested at time $t$, then our integral solution obtained in \Cref{alg:randomizedRounding} also does not evict $p_t$. It follows that we do not incur a cost at time $t$ if $z'_p(t-1)) = 0$. Otherwise, by \Cref{lem:ALGf} it holds that $z'_p(t-1) \geq \frac{1}{4 \cdot N \cdot \mu}$. 
	
	Let $c_{\min}, c_{\max}$ be the minimum and maximum costs of pages, thus $C = \frac{c_{\max}}{c_{\min}}$.  As we assume that $z'_p(t-1) \geq \frac{1}{4 \cdot N \cdot \mu}$, it follows that the cost of the solution $z'$ increases by at least $\frac{c_{\min}}{4 \cdot N \cdot \mu}$ at time $t$ by fetching the fraction of $p_t$ missing from cache. At the same time, observe that each time that the algorithm enters \Cref{step:reach}, the total cost it pays is bounded by $c_{\max} \cdot \mu$ as it can evict at most $\mu$ pages from cache as the cache has been feasible until this time.  Moreover, by \Cref{lem:Gupta}, the probability that the algorithm enters \Cref{step:reach} in time $t$ is bounded by $\frac{1}{2 \cdot C \cdot N \cdot \mu^2}$. Thus, the expected cost paid by the algorithm at time $t$ is bounded by $\frac{c_{\max} \cdot n}{2 \cdot C \cdot N \cdot \mu^2} \leq \frac{c_{\min}}{2 \cdot N \cdot \mu}$. Therefore, at each time step the expected cost of the algorithm is bounded by a constant factor of the expected cost of $z'$. In addition, the cost of $z'$ is bounded by $\alpha$ times the cost of $z$. Therefore, the proof follows from the competitive ratio of $z$ described in \Cref{lem:ALGf}. \qed

	\subsubsection*{Proof of \Cref{thm:approximation}}
	
	using the ``round-or-separate''
	approach of \cite{gupta2020online} we can solve our LP \eqref{eq:LP2} offline approximately in polynomial time with $O(\log f(\cP))$-approximation ratio. Using \Cref{lem:Gupta}  and the analysis in \Cref{lem:ALG3}, this implies the statement of \Cref{thm:approximation}.  \qed

\end{document}